\documentclass[10pt, doublecolumn]{IEEEtran}

\usepackage{amsmath,amsthm,amssymb}

\usepackage{algorithmic}
\usepackage{algorithm}

\usepackage{mathtools}
\usepackage{amsfonts}

\usepackage{graphicx}
\usepackage{subfigure}
\usepackage{float}

\usepackage{epstopdf}

\usepackage{tabularx}

\usepackage{cite}

\usepackage{hyperref}

\usepackage{cuted}
\usepackage{multicol}

\usepackage{stfloats}
\usepackage{pstricks}
\usepackage{xcolor}
\usepackage{booktabs}

\usepackage{cases}

\usepackage{url}

\allowdisplaybreaks

\newtheorem{thm}{Theorem}

\newtheorem{prop}[thm]{Proposition}
\theoremstyle{definition}

\theoremstyle{remark}
\newtheorem{rem}{Remark}

\newcounter{nestedstate}[algorithm]

\begin{document}
\newtheorem{Remark}{Remark}
\newtheorem{remark}{Remark}
\renewcommand{\algorithmicrequire}{\textbf{Input:}} 
\renewcommand{\algorithmicensure}{\textbf{Output:}}

\title{Optimal Beamforming of RIS-Aided Wireless Communications: An Alternating Inner Product Maximization Approach}

\author{Rujing~Xiong,~\IEEEmembership{Student Member,~IEEE,}
Tiebin~Mi,~\IEEEmembership{Member,~IEEE,}
Jialong~Lu,~\IEEEmembership{Student Member,~IEEE,}
Ke~Yin,
Kai~Wan,~\IEEEmembership{Member,~IEEE,}
Fuhai~Wang,~\IEEEmembership{Student Member,~IEEE,}
Robert~Caiming~Qiu,~\IEEEmembership{Fellow,~IEEE,}
\thanks{R.~Xiong, T.~Mi, J.~Lu, K.~Wan ,F.~Wang and R.~C.~Qiu are with the School of Electronic Information and Communications, Huazhong University of Science and Technology, Wuhan 430074, China (e-mail:~{\{rujing, mitiebin, m202272434, kai\_wan, wangfuhai, caiming\}@hust.edu.cn)}.}
\thanks{K.~Yin is with the Center for Mathematical Sciences, Huazhong University of Science and Technology, Wuhan 430074, China (e-mail: kyin@hust.edu.cn).}
\thanks{A short version of this paper was presented in the 2024 IEEE Wireless Communications and Networking Conference (WCNC 2024).}
}



\maketitle

\begin{abstract}

This paper investigates a general discrete $\ell_p$-norm maximization problem, with the power enhancement at steering directions through reconfigurable intelligent surfaces (RISs) as an instance. We propose a mathematically concise iterative framework composed of alternating inner product maximizations, well-suited for addressing $\ell_1$- and $\ell_2$-norm maximizations with either discrete or continuous uni-modular variable constraints. The iteration is proven to be monotonically non-decreasing. Moreover, this framework exhibits a distinctive capability to mitigate performance degradation due to discrete quantization, establishing it as the first post-rounding lifting approach applicable to any algorithm intended for the continuous solution. Additionally, as an integral component of the alternating iterations framework, we present a divide-and-sort (DaS) method to tackle the discrete inner product maximization problem. In the realm of $\ell_\infty$-norm maximization with discrete uni-modular constraints, the DaS ensures the identification of the global optimum with polynomial search complexity. We validate the effectiveness of the alternating inner product maximization framework in beamforming through RISs using both numerical experiments and field trials on prototypes. The results demonstrate that the proposed approach achieves higher power enhancement and outperforms other competitors.  Finally, we show that discrete phase configurations with moderate quantization bits (e.g., 4-bit) exhibit comparable performance to continuous configurations in terms of power gains.


\end{abstract}


\section{Introduction}\label{Section1}
\IEEEPARstart{T}{he} reconfigurable intelligent surface (RIS) technique has recently demonstrated its great potential for reconfiguring wireless propagation environments \cite{cui2014coding,basar2019wireless,di2020smart}. 
The advantage, as compared to other competitive technologies, lies in the fact that RISs provide opportunities for the so-called passive relays. Specifically, RISs consist of a large number of carefully designed electromagnetic units and result in electromagnetic waves with dynamically controllable behaviors such as amplitude, phase, and polarization. Moreover, RIS operates passively, leading to reductions in both hardware costs and energy consumption.

RIS has attracted considerable attention in various wireless systems, including multi-antenna and/or multi-user communication~\cite{di2020hybrid,chen2023channel}, physical-layer security~\cite{yang2020secrecy,cui2019secure}, orthogonal frequency division multiplexing-based wideband communication~\cite{lin2020adaptive,yang2020intelligent}, unmanned aerial vehicle communication and networks ~\cite{li2020reconfigurable,mu2021intelligent}, simultaneous wireless information and power transfer systems~\cite{pan2020intelligent,wu2020joint}, mobile edge computing~\cite{bai2020latency,mao2022reconfigurable}, wireless sensing and location~\cite{elzanaty2021reconfigurable,chen2022efficient}, etc.

In reality, achieving the proper configuration of RISs is constrained by various limitations. Among these, one remarkable hardware restriction is the adoption of extremely low quantization bits \cite{xiong2023ris,arun2020rfocus,tran2020demonstration, pei2021ris,dai2020reconfigurable,rains2022high,kishiyama2021research}. The basic control module of reflecting units typically relies on varactors or positive intrinsic-negative (PIN) diodes. The use of varactors presents a challenge due to analog control, resulting in a degradation of phase accuracy and response time. On the other hand, integrating PIN diodes in multi-bit reflecting units increases design complexity and hardware costs. As a result, the 1-bit phase configuration quantization remains the primary choice in the majority of current prototypes.

In many applications, beamforming through RIS can be formulated as maximizing the power at specific steering directions\cite{wu2019intelligent,wu2023intelligent}. The optimization of phase configurations encounters two critical constraints. The first restriction is the uni-modular phase configuration, originating from the resonant nature of electrical circuits. This constraint imposes non-convexity onto the feasible set. The second involves discrete phase configuration for practical implementations, further reducing the feasible region into a discrete set. Mathematically, this beamforming functionality represents an instance of $\ell_2$-norm maximization problem with discrete uni-modular constraints, formulated as
\begin{equation*}
  (Q_2) \qquad \max_{ \mathbf{\Omega} \in \Delta^n } \lVert \mathbf{A} e^{j \mathbf{\Omega} } \rVert_2,
\end{equation*}
where $\mathbf{A} \in \mathbb{C}^{m \times n}$ represents the channel information and $e^{j \mathbf{\Omega}} = [e^{j \Omega_1}, \cdots, e^{j\Omega_n}]^T$ denotes the phase configurations. The feasible set is denoted as $\Delta = \left\{ 0, \delta, \ldots, (2^B-1) \delta \right\}$, where $B$ represents the number of quantization bits and $\delta= 2 \pi / 2^B$.

This paper addresses a general $\ell_p$-norm maximization problem ($p=1, 2, \infty$), given by
\begin{equation*}
  (Q_p) \qquad \max_{ \mathbf{\Omega} \in \Delta^n } \lVert \mathbf{A} e^{j \mathbf{\Omega} } \rVert_p .
\end{equation*}
It is essential to recognize that discrete optimization problems $(Q_p)$ tend to fall within the universal (NP-hard) category, requiring expensive exponential search techniques. Notably, the continuous counterpart of $(Q_2)$ remains strongly NP-hard in general \cite{zhang2006complex, soltanalian2014designing}. Various methods have been developed to address the (continuous) uni-modular constraint~\cite{liang2016unimodular,yang2021beamforming,wang2022sca,kumar2022novel,wu2019intelligent,cui2019secure,zhou2020robust,shen2019secrecy,bai2021reconfigurable,salem2022active,absil2008optimization,yu2019miso}. Among these methods, semidefinite relaxation-semidefinite program (SDR-SDP) and manifold optimization (Manopt) stand out as representative approaches.

In the SDR-SDP approach~\cite{wu2019intelligent,cui2019secure,zhou2020robust}, the continuous counterpart of $(Q_2)$ is equivalent to semidefinite programming with the variable being a rank one positive semidefinite matrix. The uni-modular constraint is implemented through imposing restrictions on the diagonal elements of the matrix variable. By relaxing the rank-one constraint, the problem evolves into a well-defined semidefinite programming, which can be solved using an interior-point algorithm. It is noteworthy that while SDR-SDP can handle many nonconvex quadratically constrained quadratic programs, it is not applicable to address $(Q_1)$ and $(Q_\infty)$.

The Manopt approach~\cite{absil2008optimization,yu2019miso} deals directly with optimization over Riemannian manifolds. Most well-defined manifolds are equipped with computable retraction operators, providing an efficient method of pulling points from the tangent space back onto the manifold. In the context of $(Q_2)$, the feasible set of its continuous counterpart comprises a product of $n$ complex tori, forming an embedded submanifold of $\mathbb{C}^n$. Within the Manopt framework, there exists a straightforward way to project the Euclidean gradient $-2 A^H A x$ to the Riemannian gradient of the product manifold of $n$ tori. However, even with Manopt, there is no guarantee of finding a global minimizer due to the nonconvex nature of $(Q_2)$, as well as the nonconvex and nonsmooth characteristics of $(Q_1)$ and $(Q_\infty)$.

In accommodating the discrete configuration constraint, the most straightforward approach involves \emph{hard} rounding the solution to its continuous counterpart \cite{wang2020intelligent,zheng2020intelligent,you2020channel,qiao2020secure}. However, hard rounding may lead to performance degradation~\cite{wu2019beamforming}, and in the worst-case scenario, it can even result in arbitrarily poor performance~\cite{zhang2022configuring}. Moreover, only a limited number of algorithms have been developed to reduce the search set from exponential to polynomial complexity~\cite{di2020hybrid,zhang2022configuring,ren2022linear}, enabling the direct search of optimal discrete solutions. Discovering such algorithms can often be highly problem-specific and exceptionally challenging.

An overlooked strategy is the post-rounding lifting approach, which addresses the performance degradation caused by hard rounding. This method becomes particularly crucial when dealing with extremely low-bit phase configurations, where the loss due to hard rounding can be significant. Despite its potential effectiveness, this strategy is completely disregarded in existing literature. To our knowledge, there are no related approaches currently employing this strategy.

The main contributions can be summarized as follows:
\begin{itemize}
\item Leveraging the inherent structure of $(Q_p)$, we propose a mathematically concise iterative framework composed of alternating inner product maximizations. This framework is designed to address $\ell_1$- and $\ell_2$-norm maximization with either discrete or continuous uni-modular constraints. Within the iteration, one step guarantees the exact satisfaction of uni-modular constraints, whether they are discrete or continuous. The iteration is proven to be monotonically non-decreasing and exhibits effectiveness in handling large-scale problems. Moreover, it possesses the unique capability to mitigate performance degradation resulting from the hard rounding of continuous solutions. Remarkably, our proposed iteration turns out to be the first post-rounding lifting approach applicable to all algorithms designed for continuous solutions.


\item We introduce a divide-and-sort (DaS) method to address the discrete inner product maximization problem, which constitutes a crucial step in the proposed iterative framework \footnote{Upon revising the initial version of this paper, we became aware of recent independent research \cite{ren2022linear}, which utilizes a similar approach to DaS for solving the discrete inner product maximization problem. Notably, their study does not address the more general problem $(Q_2)$.}. The proposed DaS is guaranteed to identify the optimal discrete solution.  Furthermore, in the case of $(Q_\infty)$, the DaS ensures the identification of the global optimum with polynomial search complexity.

\item We validate the effectiveness of the alternating inner product maximization framework in beamforming using  both numerical experiments and field trials. The results demonstrate that the proposed approach achieves higher power enhancement and outperforms other competitors. Finally, we show that discrete phase configurations with moderate quantization bits (e.g., 4-bit) have comparable performance to continuous configurations in RIS-aided communication systems.

\end{itemize}


\subsection{Outline}

The remainder of the paper is organized as follows. In Section~\ref{Section2}, we demonstrate that beamforming through RISs can be formulated as an $\ell_2$-norm maximization problem. Section~\ref{Section3} presents a concise mathematical framework comprising alternating inner product maximizations to address $(Q_p)$. In Section~\ref{Section4}, we introduce a divide-and-sort (DaS) search method to address the discrete inner product maximization problem, a crucial step in the proposed iterative framework. Section~\ref{Section6} is dedicated to the convergence behaviors and unique lifting capabilities. In Sections~\ref{Section7} and~\ref{Section7}, we conduct performance evaluations through numerical simulations and field trials. The paper is concluded in Section~\ref{Section9}.

\subsection{Reproducible Research}
The simulation results can be reproduced using codes available at: \url{https://github.com/RujingXiong/RIS\_Optimization.git}

\subsection{Notations}
The imaginary unit is denoted by $j$. The magnitude and real component of a complex number are represented by $|\cdot|$ and $\mathcal{\Re}\{\cdot\}$, respectively. Unless explicitly specified, lower and upper case bold letters represent vectors and matrices. The conjugate transpose, conjugate, and transpose of~$\mathbf{A}$ are denoted as $\mathbf{A}^H$, $\mathbf{A}^{*}$ and $\mathbf{A}^T$, respectively. ${\rm diag}(\cdot)$ refers to the diagonal matrix operator. 

\section{System Model and Beamforming Problem Formulation}\label{Section2}
We investigate a point-to-point multiple-input single-output (MISO) downlink communication system as illustrated in Fig.~\ref{Scenarios}, 
where a base station (BS) equipped with $M$ antennas serves a single-antenna user equipment (UE). In this scenario, linear transmit precoding is applied at the BS, allocating the UE a dedicated beamforming vector $\mathbf{w}$. The transmitted signal is represented as $\mathbf{x} = \mathbf{w} s$, where $s$ is the transmitted symbol.

To enhance communication quality, a RIS consisting of $N$ reflecting units is employed. By representing the phase configuration of the RIS as $e^{j \mathbf{\Omega}} = [e^{j \Omega_1}, \cdots, e^{j\Omega_N}]^T$, the received signal at the UE is given by
\begin{equation}\label{modelingLOS}
y = ( \mathbf{h}_{\text{UE-RIS}}^H \text{diag} (e^{j \mathbf{\Omega}} ) \mathbf{H}_{\text{RIS-BS}}+\mathbf{h}_d^H ) \mathbf{w} s + v,
\end{equation}
where the equivalent channel between RIS-BS and UE-RIS links are denoted by $\mathbf{H}_{\text{RIS-BS}} \in \mathbb{C}^{N \times M}$ and $\mathbf{h}_{\text{UE-RIS}}^H \in \mathbb{C}^{N}$, respectively. The direct link between the BS and the UE is represented by $\mathbf{h}_d^H\in \mathbb{C}^M$. The received signal is corrupted by additive Gaussian white noise $v \sim \mathcal{CN}(0,\sigma^2)$.

\begin{figure}[tbp]
  \centering
  \subfigure[]{
  \label{Scenarios:LOS}
  \includegraphics[width=.75\columnwidth]{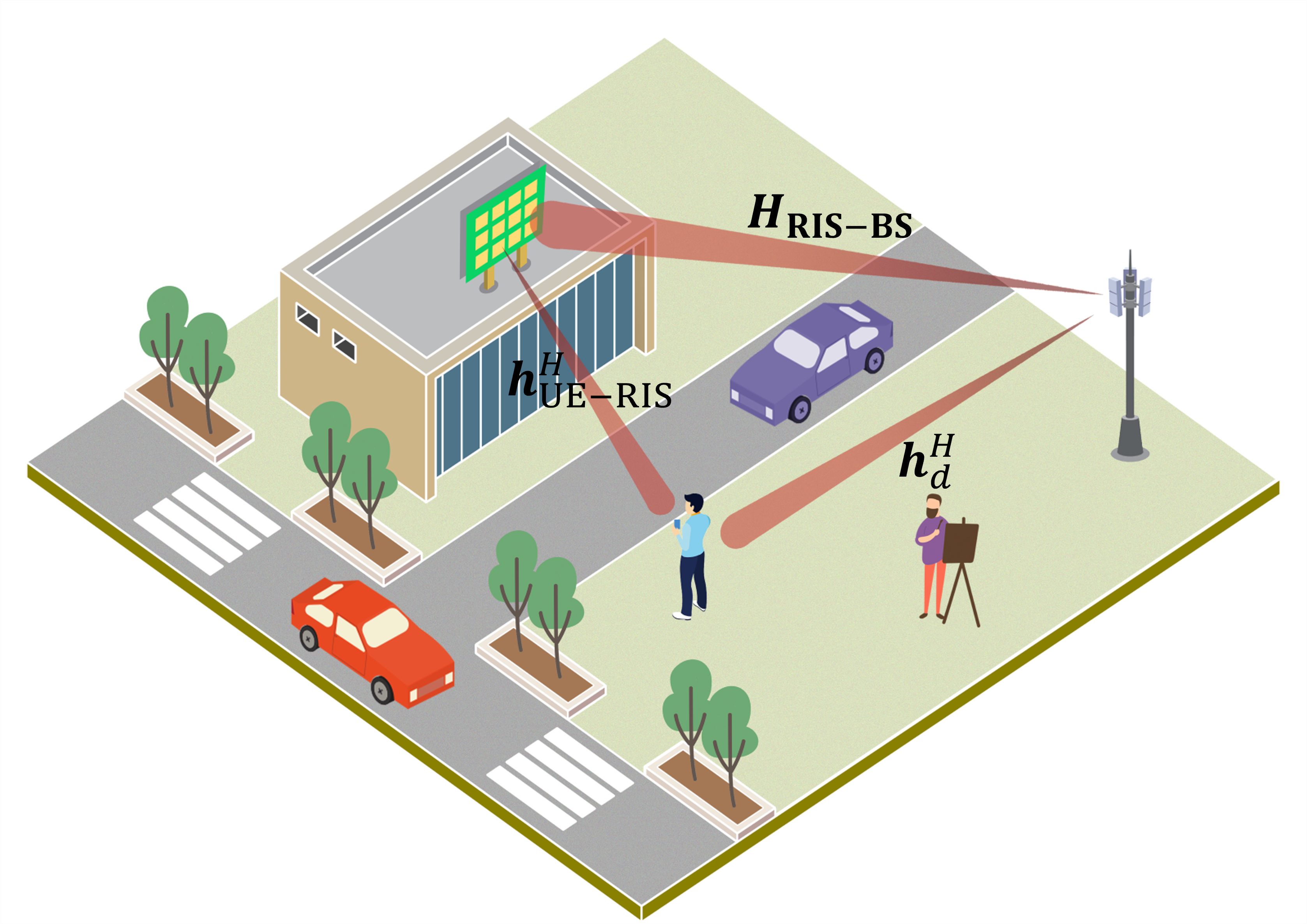}}
  \subfigure[]{
  \label{Scenarios:NLOS}
  \includegraphics[width=.75\columnwidth]{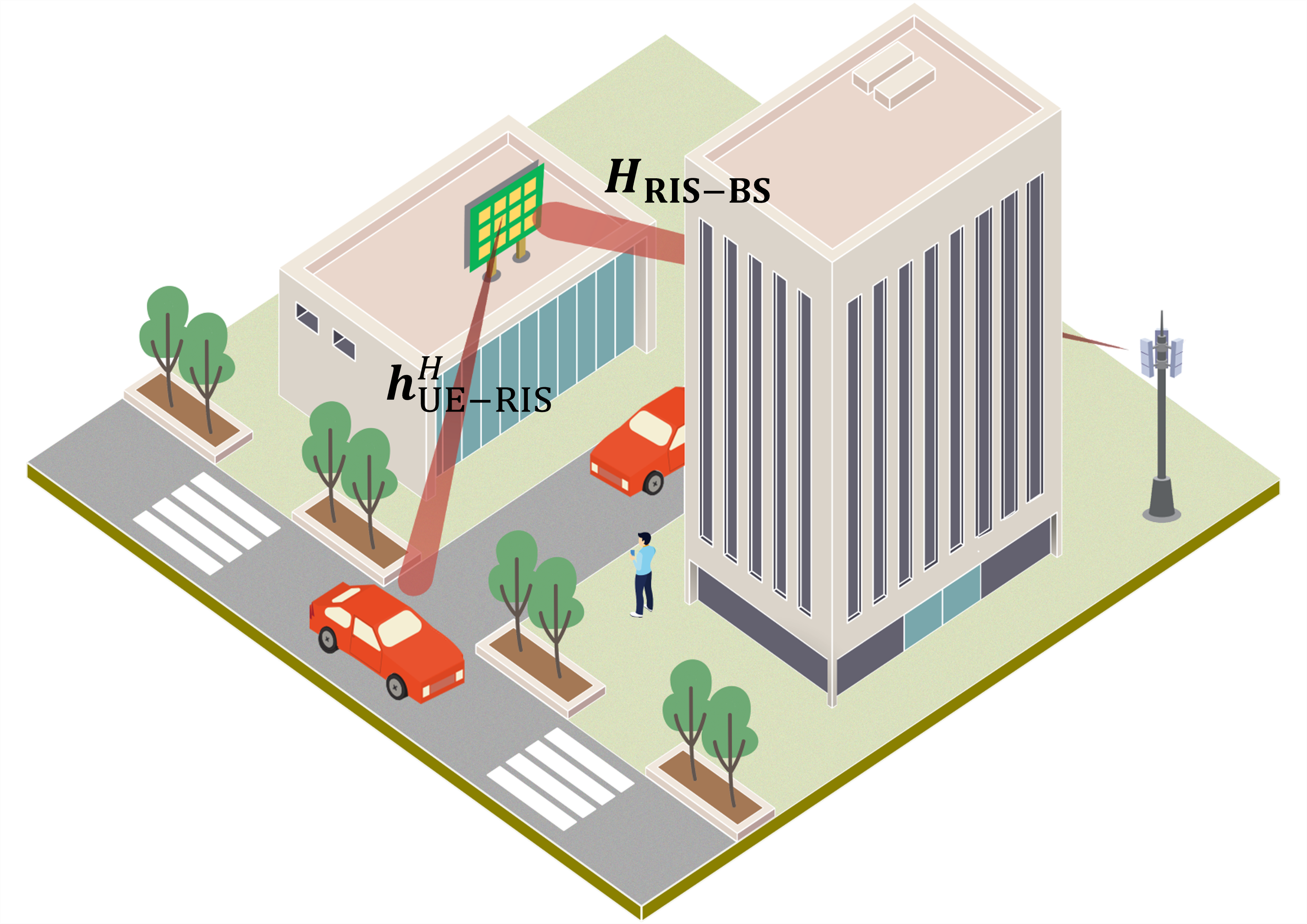}}
    \caption{RIS-aided point-to-point MISO communication. (a) line-of-sight (LoS) scenario. (b) non-line-of-sight (NLoS) scenario.} 
  \label{Scenarios}
\end{figure}

The objective of RIS-aided coverage enhancement is to maximize the received signal power. According to maximum ratio transmission principle \cite{wu2019intelligent, tse2005fundamentals}, the optimal transmit beamforming vector is given by
\begin{equation}
  \mathbf{w}= \sqrt{P} \frac{\left( \mathbf{h}_{\text{UE-RIS}}^H \text{diag} (e^{j \mathbf{\Omega}} ) \mathbf{H}_{\text{RIS-BS}}+\mathbf{h}_d^H \right)^H }{ \lVert \mathbf{h}_{\text{UE-RIS}}^H \text{diag} (e^{j \mathbf{\Omega}} ) \mathbf{H}_{\text{RIS-BS}}+\mathbf{h}_d^H \rVert},
\end{equation}
where $P$ denotes the transmit power of the BS. Employing this optimal transmit beamforming vector, the problem of maximizing the received  power can be formulated as
\begin{equation*}
 \max_{ \mathbf{\Omega} } \lVert \mathbf{h}_{\text{UE-RIS}}^H \text{diag} (e^{j \mathbf{\Omega}} ) \mathbf{H}_{\text{RIS-BS}}+\mathbf{h}_d^H \rVert_2 .
\end{equation*}
In addition, since
\begin{equation}\label{E:Phi}
  \mathbf{h}_{\text{UE-RIS}}^H \text{diag} ( e^{j \mathbf{\Omega}} ) \mathbf{H}_{\text{RIS-BS}} = e^{j \mathbf{\Omega}} \underbrace{ \text{diag} ( \mathbf{h}_{\text{UE-RIS}}^H ) \mathbf{H}_{\text{RIS-BS}} }_{ = \mathbf{\Phi} } ,
\end{equation}
we can rewrite the received signal as
\begin{equation}
  \mathbf{h}_{\text{UE-RIS}}^H \text{diag} (e^{j \mathbf{\Omega}} ) \mathbf{H}_{\text{RIS-BS}}+\mathbf{h}_d^H = 
  \begin{bmatrix} 
    \mathbf{\Phi}^T & \mathbf{h}_d^H 
  \end{bmatrix} 
  \begin{bmatrix}
    e^{j \mathbf{\Omega}} \\
    1
  \end{bmatrix} .
\end{equation}

In practical implementations, the phase configuration is typically selected from a finite set of values $\Delta = \left\{ 0, \delta, \ldots, (2^B-1) \delta \right\}$, where $B$ represents the number of bits of quantization and $\delta= 2 \pi / 2^B$. The optimal beamforming can be  expressed as
\begin{equation}\label{P3}
  \max_{ \mathbf{\Omega} \in \Delta^N } \left \lVert 
  \begin{bmatrix} 
    \mathbf{\Phi}^T & \mathbf{h}_d^H 
  \end{bmatrix} 
  \begin{bmatrix}
    e^{j \mathbf{\Omega}} \\
    1
  \end{bmatrix} \right \rVert_2 .
\end{equation}
For a more general optimization problem 
\begin{equation}\label{P4}
\max_{ \mathbf{\Omega}, \Omega_{N+1} \in \Delta^{N+1} } \left \lVert 
  \begin{bmatrix} 
    \mathbf{\Phi}^T & \mathbf{h}_d^H 
  \end{bmatrix} 
  \begin{bmatrix}
    e^{j \mathbf{\Omega}} \\
    e^{j \Omega_{N+1}}
  \end{bmatrix} \right \rVert_2 .
\end{equation}
We denote the solution as $\mathbf{\Omega}_{\text{opt}}, \Omega_{N+1, \text{opt}}$. Consequently, the optimal solution for \eqref{P3} can be represented as $\mathbf{\Omega}_{\text{opt}} - \Omega_{N+1, {\text{opt}}}$. Indeed, \eqref{P4} is a specific instance of a broader optimization problem $(Q_2)$.

In certain scenarios, as depicted in Fig.~\ref{Scenarios:NLOS}, the link between the BS and the UE is completely blocked. In such non-line-of-sight (NLOS) conditions, employing RIS to enhance signal coverage leads to an optimization problem formulated as
\begin{equation}\label{Optimization Problem}
\max_{ \mathbf{\Omega} \in \Delta^N } \lVert \mathbf{\Phi }^T e^{j \mathbf{\Omega} } \rVert_2,
\end{equation}
which is evidently an instance of $(Q_2)$.

\section{An Alternating Maximization Framework}\label{Section3}

The $\ell_p$-norm maximization $(Q_p)$ with discrete uni-modular constraints holds fundamental importance in various applications, particularly in engineering problems related to quantization. One approach to solving $(Q_p)$ is to perform an exhaustive search among all elements in $\Delta^{n}$, with an intractable computation complexity of $\mathcal{O}(2^{B n})$. In order to reduce the complexity, we propose a mathematically concise iterative framework consisting of alternating inner product maximizations. This framework is specifically designed to address $\ell_1$- and $\ell_2$-norm maximization problems with either discrete or continuous uni-modular constraints.

Let $\lVert \cdot \rVert_q$ represent $\ell_q$-norm ($q \ge 1$) on $\mathbb{C}^{n}$ and $\langle \cdot, \cdot \rangle$ denote the inner product. The associated dual norm $\lVert \cdot \rVert_{p}$ is denoted as 
\begin{equation}\label{E:DualNorm}
  \lVert \mathbf{x} \rVert_{p} = \sup \{ \lvert \langle \mathbf{z}, \mathbf{x} \rangle \rvert : \lVert \mathbf{z} \rVert_q = 1 \} ,
\end{equation}
where $1/p + 1/q = 1$. By setting $p=2$ and $q=2$, we know that the dual of the Euclidean norm is the Euclidean norm, given by
\begin{equation}
  \lVert \mathbf{x} \rVert_{2} = \sup \{ \lvert \mathbf{z}^H \mathbf{x} \rvert : \lVert \mathbf{z} \rVert_2 = 1 \} .
\end{equation}

To reformulate the original maximization $(Q_p)$, we introduce a dual variable $\{ \mathbf{z} : \lVert \mathbf{z} \rVert_q = 1 \}$. This yields
\begin{equation}
\begin{aligned}
\max_{ \mathbf{\Omega} \in \Delta^n } \lVert \mathbf{A} e^{j \mathbf{\Omega} } \rVert_p 
& = \max_{ \mathbf{\Omega} \in \Delta^n } \sup_{ \lVert \mathbf{z} \rVert_q = 1 } \lvert \langle \mathbf{z} , \mathbf{A} e^{ j \mathbf{\Omega} } \rangle \rvert \\
& = \sup_{ \lVert \mathbf{z} \rVert_q = 1 } \max_{ \mathbf{\Omega} \in \Delta^n } \lvert \langle \mathbf{A}^H \mathbf{z} , e^{ j \mathbf{\Omega} } \rangle \rvert.
\end{aligned}
\end{equation}
Our objective becomes solving
\begin{equation}\label{Maxmax_1}
(\mathbf{z}_{\text{opt}}, \mathbf{\Omega}_{\text{opt}}) = \arg \max_{ \mathbf{\Omega} \in \Delta^n } \sup_{ \lVert \mathbf{z} \rVert_q = 1 } \lvert \langle \mathbf{z} , \mathbf{A} e^{ j \mathbf{\Omega} } \rangle \rvert
\end{equation}
or equivalently
\begin{equation}\label{Maxmax_2}
(\mathbf{z}_{\text{opt}}, \mathbf{\Omega}_{\text{opt}}) = \arg \sup_{ \lVert \mathbf{z} \rVert_q = 1 } \max_{ \mathbf{\Omega} \in \Delta^n } \lvert \langle \mathbf{A}^H \mathbf{z} , e^{ j \mathbf{\Omega} } \rangle \rvert.
\end{equation}

The reformulation appears to introduce more complexity as \eqref{Maxmax_1} and \eqref{Maxmax_2} expand the feasible set from $\Delta^n$ to $\Delta^n \times \{ \lVert \cdot \rVert_q = 1 \}$. However, the advantage lies in the fact that we can use alternating restricted maximizations over the subsets of variables. For instance, in \eqref{Maxmax_1}, if $\mathbf{\Omega}_{\text{opt}}$ is known, $\mathbf{z}_{\text{opt}}$ is the solution to a continuous inner product maximization problem
\begin{equation}
  \sup_{ \lVert \mathbf{z} \rVert_q = 1 } \lvert \langle \mathbf{z} , \mathbf{A} e^{ j \mathbf{\Omega}_{\text{opt}} } \rangle \rvert.
\end{equation}
On the other hand, if $\mathbf{z}_{\text{opt}}$ is given, $\mathbf{\Omega}_{\text{opt}}$ is the solution to a discrete uni-modular constrained inner product maximization problem  
\begin{equation}
  \max_{ \mathbf{\Omega} \in \Delta^n } \lvert \langle \mathbf{z}_{\text{opt}} , \mathbf{A} e^{ j \mathbf{\Omega} } \rangle \rvert = \max_{ \mathbf{\Omega} \in \Delta^n } \lvert \langle \mathbf{A}^H \mathbf{z}_{\text{opt}} , e^{ j \mathbf{\Omega} } \rangle \rvert.
\end{equation}
We now reach an iterative procedure involving alternating inner product maximization, defined by the following steps 
\begin{subnumcases}
  \mathbf{z}_k = \arg \sup_{ \lVert \mathbf{z} \rVert_q = 1 } \lvert \langle \mathbf{z} , \mathbf{A} e^{ j \mathbf{\Omega}_k } \rangle \rvert, \label{E:Arg_z} \\
  \mathbf{\Omega}_{k+1} = \arg \max_{ \mathbf{\Omega} \in \Delta^n } \lvert \langle \mathbf{A}^H \mathbf{z}_{k} , e^{ j \mathbf{\Omega} } \rangle \rvert . \label{E:Arg_Omega}
\end{subnumcases}

The following theorem investigates the convergence inherent in this iterative sequence.

\begin{thm}\label{T:Iteration}
Given any initial $\mathbf{\Omega}_0 \in \Delta^n$, let $\mathbf{\Omega}_0, \mathbf{\Omega}_1, \ldots$ be the sequence obtained from \eqref{E:Arg_z} and \eqref{E:Arg_Omega}. The following inequality holds for the sequence of $\ell_p$-norms
\begin{equation}
  \lVert \mathbf{A} e^{j \mathbf{\Omega}_0 } \rVert_p \le \lVert \mathbf{A} e^{j \mathbf{\Omega}_1 } \rVert_p \le \cdots \le \lVert \mathbf{A} e^{j \mathbf{\Omega}_k } \rVert_p \le \cdots .
\end{equation}
\end{thm}

\begin{proof}
  By utilizing the definitions presented in \eqref{E:DualNorm} and \eqref{E:Arg_z}, we can express $\lVert \mathbf{A} e^{j \mathbf{\Omega}_k } \rVert_p = \lvert \langle \mathbf{z}_k , \mathbf{A} e^{ j \mathbf{\Omega}_k } \rangle \rvert = \lvert \langle \mathbf{A}^H \mathbf{z}_k , e^{ j \mathbf{\Omega}_k } \rangle \rvert$. Further leveraging \eqref{E:Arg_Omega}, we establish $ \lvert \langle \mathbf{A}^H \mathbf{z}_k , e^{ j \mathbf{\Omega}_k } \rangle \rvert \le \lvert \langle \mathbf{A}^H \mathbf{z}_k , e^{ j \mathbf{\Omega}_{k+1} } \rangle \rvert = \lvert \langle \mathbf{z}_k , \mathbf{A} e^{ j \mathbf{\Omega}_{k+1} } \rangle \rvert$. Furthermore, from \eqref{E:Arg_z}, we derive that $\lvert \langle \mathbf{z}_k , \mathbf{A} e^{ j \mathbf{\Omega}_{k+1} } \rangle \rvert \le \lvert \langle \mathbf{z}_{k+1} , \mathbf{A} e^{ j \mathbf{\Omega}_{k+1} } \rangle \rvert = \lVert \mathbf{A} e^{j \mathbf{\Omega}_{k+1} } \rVert_p$. We finally establish the inequality $\lVert \mathbf{A} e^{j \mathbf{\Omega}_k } \rVert_p \le \lVert \mathbf{A} e^{j \mathbf{\Omega}_{k+1} } \rVert_p$.
\end{proof}

\begin{rem}\label{R:lifting}
Theorem~\ref{T:Iteration} establishes that for any initial point $\mathbf{\Omega}_0 \in \Delta^n$, the iterative procedure generates a sequence $\mathbf{\Omega}_1, \mathbf{\Omega}_2, \ldots$ where the $\ell_p$-norms (cost function) exhibit a monotonically non-decreasing behavior. The role of \eqref{E:Arg_Omega} is to ensure that the sequence adheres to the discrete uni-modular constraints. The proposed iteration is notable for its ability to effectively compensate for the performance degradation caused by hard rounding. This distinctive feature positions the proposed method as the first post-rounding \emph{lifting} approach applicable to any algorithms with continuous solutions. 
\end{rem}

The following sections focus on solving both the continuous and discrete inner product maximization problems presented in \eqref{E:Arg_z} and \eqref{E:Arg_Omega}. Initially, we address the continuous inner product maximization problem.

\subsection{Analytical Solutions for Continuous Inner Product Maximization Problems}
For a given vector $\mathbf{w} \in \mathbb{C}^{n}$, our objective is to solve 
\begin{equation}
  \sup_{ \lVert \mathbf{z} \rVert_q = 1 } \lvert \langle \mathbf{z} , \mathbf{w} \rangle \rvert.  
\end{equation}
Fortunately, analytical solutions exist for $q=2, \infty$ (corresponding to $p=2, 1$, respectively). Specifically, when $q = 2$ ($p=2$), for  $\lVert \mathbf{z} \rVert_2 = 1$, $\lvert \langle \mathbf{z} , \mathbf{w} \rangle \rvert \le \lVert \mathbf{z} \rVert_2 \lVert \mathbf{w} \rVert_2 = \lVert \mathbf{w} \rVert_2$. By setting $\mathbf{z} = e^{j \phi} \mathbf{w} / \lVert \mathbf{w} \rVert_2$, we achieve  $\lvert \langle \mathbf{z} , \mathbf{w} \rangle \rvert = \lVert \mathbf{w} \rVert_2$. Here $\phi$ is an arbitrary angle. For simplicity, in the case of $q=2$, we choose $\mathbf{z}^k = \mathbf{A} e^{ j \mathbf{\Omega}_k } / \lVert \mathbf{A} e^{ j \mathbf{\Omega}_k } \rVert_2$ in \eqref{E:Arg_z}.

We now consider the scenario where $q = \infty$ ($p=1$). H{\"o}lder's inequality asserts that, for $\lVert \mathbf{z} \rVert_\infty = 1$, $\lvert \langle \mathbf{z} , \mathbf{w} \rangle \rvert \le \lVert \mathbf{z} \rVert_\infty \lVert \mathbf{w} \rVert_1 = \lVert \mathbf{w} \rVert_1$. By setting $\mathbf{z} = e^{j \phi} e^{j \angle \mathbf{w}}$, where $\phi$ is also an arbitrary angle, we achieve  $\lvert \langle \mathbf{z} , \mathbf{w} \rangle \rvert = \lVert \mathbf{w} \rVert_1$. Similarly, for the case $q = \infty$, we choose $\mathbf{z}^k = e^{ j \angle \mathbf{A} e^{ j \mathbf{\Omega}_k } }$ in \eqref{E:Arg_z}. Thus, we successfully obtain the solution for the scenarios where $q=2, \infty$ (corresponding to $p=2, 1$).

\begin{rem}
When $q=1$ ($p=\infty$), an analytical solution also exists for the problem $\sup_{ \lVert \mathbf{z} \rVert_1 = 1 } \lvert \langle \mathbf{z} , \mathbf{w} \rangle \rvert$. However, we do not employ the iteration presented in \eqref{E:Arg_z} and \eqref{E:Arg_Omega} to address $(Q_\infty)$. Instead, we develop an efficient divide-and-sort (DaS) method to solve the discrete inner product maximization problem \eqref{E:Arg_Omega}. In the case of $(Q_\infty)$, the DaS enables the identification of the global optimum with polynomial search complexity. 
\end{rem}

Before exploring the DaS method, we consider the selection of initial starting point in (15). We replace \eqref{E:Arg_Omega} with its continuous counterpart, leading to the following iteration
\begin{subnumcases}
  \mathbf{z}_k = \arg \sup_{ \lVert \mathbf{z} \rVert_q = 1 } \lvert \langle \mathbf{z} , \mathbf{A} e^{ j \mathbf{\Omega}_k } \rangle \rvert , \label{E:Arg_z_2} \\
  \mathbf{\Omega}_{k+1} = \arg \sup_{ \mathbf{\Omega} } \lvert \langle \mathbf{A}^H \mathbf{z}_{k} , e^{ j \mathbf{\Omega} } \rangle \rvert . \label{E:Arg_Omega_2}
\end{subnumcases}
Analogous to Theorem~\ref{T:Iteration}, this iterative process generates a sequence $\mathbf{\Omega}_1, \mathbf{\Omega}_2, \ldots$ where the $\ell_p$-norm sequence $\lVert \mathbf{A} e^{j \mathbf{\Omega}_1 } \rVert_p, \lVert \mathbf{A} e^{j \mathbf{\Omega}_2 } \rVert_p, \ldots$ increases monotonically. Indeed, the convergence of such iterations gives rise to a local maxima for the problem
\begin{equation}
\max_{ \mathbf{\Omega} } \lVert \mathbf{A} e^{j \mathbf{\Omega} } \rVert_p.
\end{equation}
In practice, this local maxima serves as an excellent starting point for the iteration presented in \eqref{E:Arg_z} and \eqref{E:Arg_Omega}.

Solving \eqref{E:Arg_Omega_2} is straightforward. Following the analysis for the case $q = \infty$ ($p=1$), H{\"o}lder's inequality asserts that $\lvert \langle \mathbf{A}^H \mathbf{z}_{k} , e^{ j \mathbf{\Omega} } \rangle \rvert \le \lVert \mathbf{A}^H \mathbf{z}_{k} \rVert_1 \lVert e^{ j \mathbf{\Omega} } \rVert_\infty = \lVert \mathbf{A}^H \mathbf{z}_{k} \rVert_1$. By setting $\mathbf{\Omega} = e^{j \phi} \angle \mathbf{A}^H \mathbf{z}_{k}$, we achieve $\lvert \langle \mathbf{A}^H \mathbf{z}_{k} , e^{ j \mathbf{\Omega} } \rangle \rvert = \lVert \mathbf{A}^H \mathbf{z}_{k} \rVert_1$. Here $\phi$ represents an arbitrary angle. In practical, we can choose $\mathbf{\Omega}_{k+1} = \angle \mathbf{A}^H \mathbf{z}_{k}$ for simplicity.

\section{A General Divide-and-sort Search Framework for Discrete Inner Product Maximization}\label{Section4}

In this section, we focus on the discrete uni-modular constrained inner product maximization problem \eqref{E:Arg_Omega}, which is a specific instance of a broader problem as
\begin{equation}\label{P1}
\max_{ \mathbf{\Omega} \in \Delta^n } \lvert \langle \mathbf{v} , e^{ j \mathbf{\Omega} } \rangle \rvert .
\end{equation}
Here $\mathbf{v} \in \mathbb{C}^n$ represents a given vector.

By introducing an auxiliary variable $\psi$ representing the argument, we can express the absolute value as 
\begin{equation}
   \lvert \langle \mathbf{v} , e^{ j \mathbf{\Omega} } \rangle \rvert = \max_{\psi \in [0, 2 \pi)} \Re \left \{ e^{- j \psi} \mathbf{v}^H e^{j \mathbf{\Omega} } \right \} .
\end{equation}
With the polar form $v_i = \lvert v_i \rvert e^{j\tau_i}$, we have
\begin{equation}\label{E:max_max_2}
  \begin{aligned}
   & \max_{ \mathbf{\Omega} \in \Delta^n } \lvert \langle \mathbf{v} , e^{ j \mathbf{\Omega} } \rangle \rvert  \\
 = & \max_{ \mathbf{\Omega} \in \Delta^n } \max_{\psi \in [0, 2 \pi)}  \Re \left \{ \sum_{i=1}^{n} \lvert v_i \rvert e^{j ( - \psi + \tau_i + \Omega_i ) } \right \} \\
 = & \max_{\psi \in [0, 2 \pi)} \max_{ \mathbf{\Omega} \in \Delta^n }   \Re \left \{ \sum_{i=1}^{n} \lvert v_i \rvert e^{j ( - \psi + \tau_i + \Omega_i ) } \right \} .
  \end{aligned}
\end{equation}
The order of the variables being maximized does not impact the results. Thus, we successfully identify an equivalent form to \eqref{P1}
\begin{equation}\label{E:EquivalentForm}
  \max_{\psi \in [0, 2 \pi)} \sum_{i=1}^{n} \left( \lvert v_i \rvert \max_{ \Omega_i \in \Delta } \cos \left( \psi - (\tau_i + \Omega_i) \right) \right) .
\end{equation}
The advantage of \eqref{E:EquivalentForm} is that it yields a search set, parameterized by $\psi$, that contains the optimal solution. Importantly, this search set has a significantly smaller cardinality compared to the exhaustive search space.

\subsection{Parameterized Search Set}
To construct a parameterized search set with a smaller cardinality, we first consider a sub-problem defined as
\begin{align}
  \Omega_{i, \text{opt}} (\psi) \equiv \arg \max_{ \Omega_i \in \Delta } \cos \left( \psi - (\tau_i + \Omega_i) \right) .
  \label{E:Subproblem}
\end{align}
For a given $\psi$, the value of $\cos \left( \psi - (\tau_i + \Omega_i) \right)$ increases as the angular difference $| \psi - (\tau_i + \Omega_i) |$ decreases. The primary task is to select $\Omega_i$ that minimizes $| \psi - (\tau_i + \Omega_i) |$. 

\begin{figure}[htbp]
  \centering
  \subfigure[]{
  \label{Fig:Solutions:SingleElement}
  \includegraphics[width=.7\columnwidth]{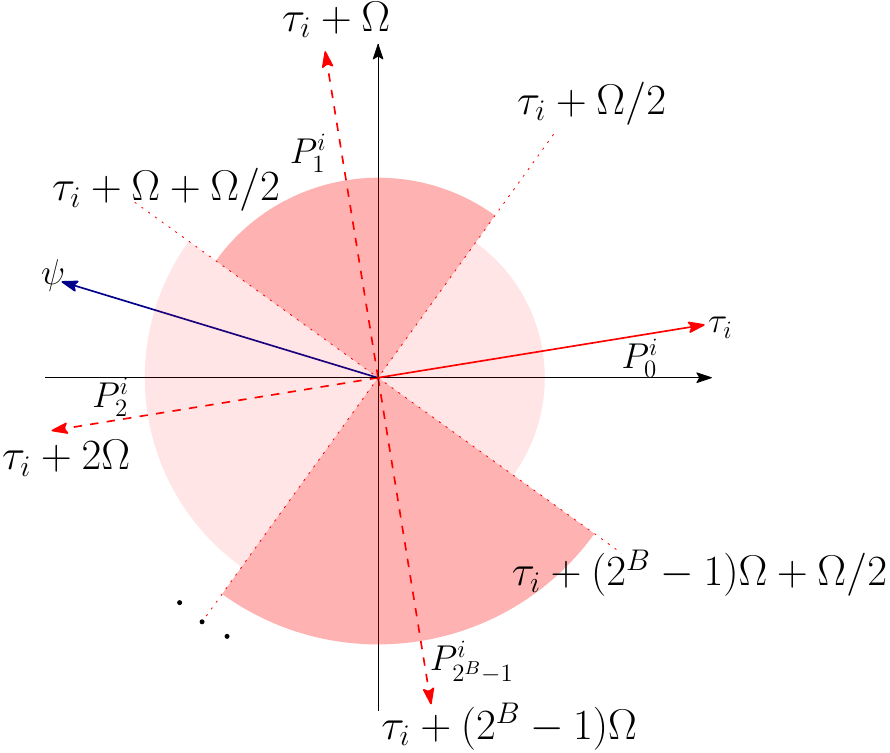}}
  \subfigure[]{
  \label{Fig:Solutions:Piecewise}
  \includegraphics[width=.8\columnwidth]{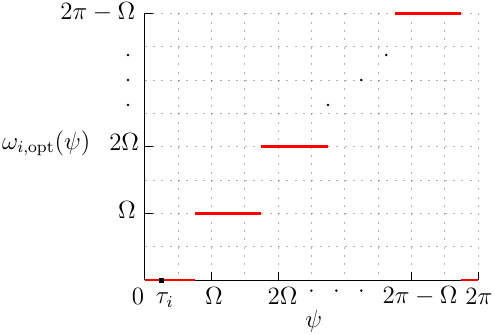}}
  \caption{The approach to find a suitable solution to \eqref{E:Subproblem}. (a) When $\psi$ is given and held fixed, the step to determine the best phase configuration $\Omega_{i, \text{opt}} (\psi)$ is to check which region $\psi$ belongs to. (b) The piecewise constant nature of the function  $\Omega_{i, \text{opt}} (\psi)$.}
  \label{Fig:Solutions}
\end{figure}

In fact, $\tau_i + \Omega_i$ has $2^B$ potential values. We divide the interval $[0, 2 \pi)$ into $2^B$ equal-length regions $\mathbb{P}^i_0, \mathbb{P}^i_1, \dots, \mathbb{P}^i_{2^B-1}$. Each region $\mathbb{P}^i_{n}= \left[ ( \tau_i + n \delta ) - \delta / 2, ( \tau_i + n \delta ) + \delta / 2 \right)$ is centered around $\tau_i + n\delta$, as illustrated in Fig.~\ref{Fig:Solutions:SingleElement}. When $\psi$ is given, determining the optimal $\Omega_{i, \text{opt}} (\psi)$ is straightforward: it corresponds to the region to which $\psi$ belongs. Specifically, if $ ( \tau_i + n \delta ) - \delta / 2 \le \psi < ( \tau_i + n \delta ) + \delta / 2 $, then the optimal configuration $\Omega_{i, \text{opt}} (\psi) = n \delta$. Fig.~\ref{Fig:Solutions:SingleElement} provides a visual representation on the ingredients to find such $\Omega_{i, \text{opt}} (\psi)$. By letting $\psi$ vary within the interval $[0, 2\pi)$, we can generate a plot of the function $\Omega_{i, \text{opt}} (\psi)$, as demonstrated in Fig.~\ref{Fig:Solutions:Piecewise}. One notable characteristic of this function is its piecewise constant nature, where $\Omega_{i, \text{opt}} (\psi)$ remains constant within each region $\mathbb{P}^i_{n}$.

In \eqref{E:EquivalentForm}, there are $n$ inner maximization sub-problems, each corresponding to an instance of \eqref{E:Subproblem}. The solutions for all sub-problems can be represented as a vector $\mathbf{\Omega}_{\text{opt}} (\psi)$. We observe that there exist $n$ groups of partitions, with each group consisting of $2^B$ regions. By rearranging the boundaries of these regions, the angle $[0, 2\pi)$ can be divided into $n 2^B$ non-overlapping regions, denoted as $\mathbb{P}^1, \ldots, \mathbb{P}^{n 2^B}$. We thus build a new search set, as the problem \eqref{E:EquivalentForm} can be rewritten as 
\begin{equation}
  \max_{\psi \in \mathbb{P}^1 \cup \cdots \cup \mathbb{P}^{n 2^B } } \sum_{i=1}^{n} \left( \lvert v_i \rvert \max_{ \Omega_i \in \Delta } \cos \left( \psi - (\tau_i + \Omega_i) \right) \right) .
\end{equation}
The advantage of such rearrangement lies in the fact that, for $\psi \in \mathbb{P}^l$, the solution vector $\mathbf{\Omega}_{\text{opt}} (\psi)$ to the inner maximization problems remains constant.

We provide an illustration for the rearrangement in Fig.~\ref{Fig:MultipleElements}. Two groups of partitions are represented by red (on the left) and blue arcs (in the middle), respectively. After rearranging, they combine to form $2 \times 2^B$ non-overlapping regions (denoted by $[\mathbb{P}_{0}^1,\mathbb{P}_{0}^2],[\mathbb{P}_{1}^1,\mathbb{P}_{0}^2]\cdots, [\mathbb{P}_{2^B-1}^1,\mathbb{P}_{2^B-1}^2] $), as shown on the right in Fig.~\ref{Fig:MultipleElements}. Within each region, the value of $\mathbf{\Omega}_{\text{opt}} (\psi)$ remains constant.

\begin{figure*}
  \centering
  \includegraphics[width=0.9\linewidth]{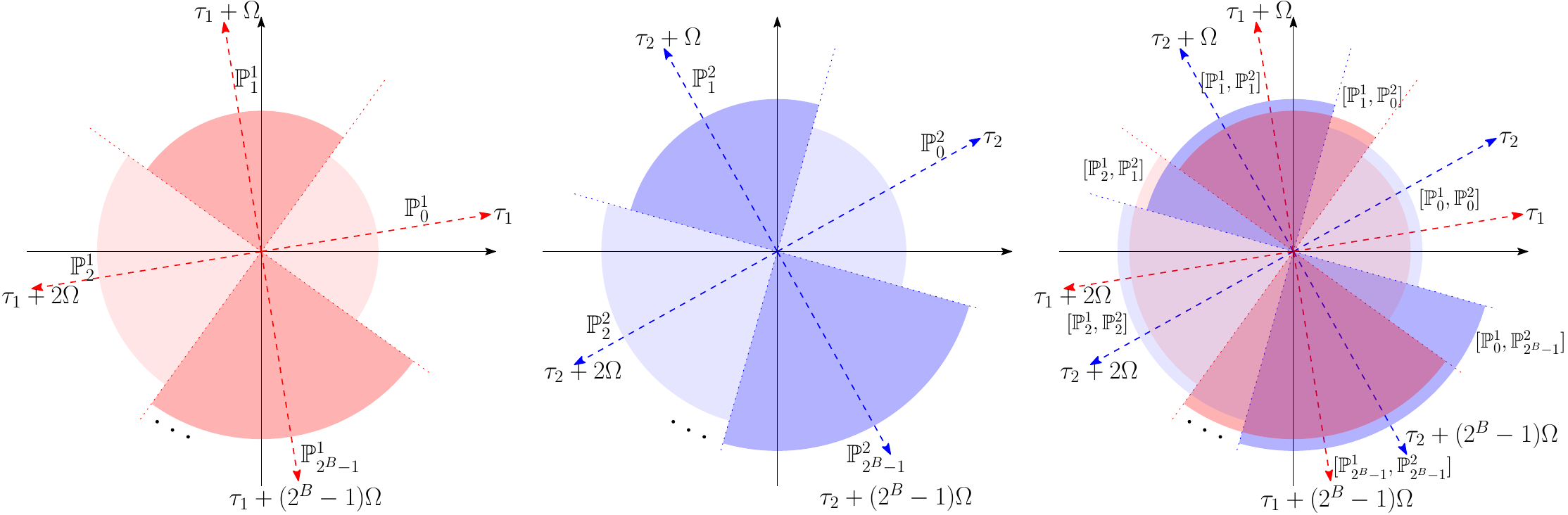}
  \caption{Illustration of region rearrangement. There are two groups of partitions ($n=2$), represented by red (on the left) and blue arcs (in the middle), each containing $2^B$ regions. After the rearrangement, these two groups combine to form a total of $2 \times 2^B$ non-overlapping regions (on the right).}
  \label{Fig:MultipleElements}
\end{figure*}

\subsection{Efficient Encoding for $n 2^B$ Regions}

We now explore encoding schemes for all regions $\mathbb{P}^1, \ldots, \mathbb{P}^{n 2^B}$. Since $\tau_1, \ldots, \tau_N$ are known \emph{a priori}, we can represent the centers of all regions as
\begin{equation}\label{E:CandidateSet}
\begin{array}{cl}
\mathcal{C}^1 & = \left\{ \tau_1, \tau_1 + \delta, \cdots, \tau_1 + (2^B-1) \delta \right\} \mod 2 \pi, \\
& \vdots \\
\mathcal{C}^n & = \left\{ \tau_n, \tau_n + \delta, \cdots, \tau_n + (2^B-1) \delta \right\} \mod 2 \pi.
\end{array}
\end{equation}
Without loss of generality, we assume that, for each $i$, the smallest center in $\mathcal{C}^i$ is $\tau_i$. We establish a bijective mapping $\mathcal{M}$ from $\{1, \ldots, n\}$ to $\{1, \ldots, n\}$ such that 
\begin{equation}\label{E:BijectiveMapping}
\tau_{\mathcal{M}(1)} \le \tau_{\mathcal{M}(2)} \le \cdots \le \tau_{\mathcal{M}(n)}.
\end{equation}
Each region can be encoded as
\begin{equation}\label{E:EncodingPartition}
\begin{aligned}
\begin{array}{cl}
  \mathbb{P}^1 & := [\tau_{\mathcal{M}(1)}, \tau_{\mathcal{M}(2)}, \cdots, \tau_{\mathcal{M}(n)}]^T + \delta [1, 0, \cdots, 0]^T, \\
  & \vdots \\
  \mathbb{P}^{n} & := [\tau_{\mathcal{M}(1)}, \tau_{\mathcal{M}(2)}, \cdots, \tau_{\mathcal{M}(n)}]^T + \delta [1, 1, \cdots, 1]^T, \\
  \mathbb{P}^{n+1} & := [\tau_{\mathcal{M}(1)}, \tau_{\mathcal{M}(2)}, \cdots, \tau_{\mathcal{M}(n)}]^T + \delta [2, 1, \cdots, 1]^T, \\
  & \vdots \\
  \mathbb{P}^{2 n} & := [\tau_{\mathcal{M}(1)}, \tau_{\mathcal{M}(2)}, \cdots, \tau_{\mathcal{M}(n)}]^T + \delta [2, 2, \cdots, 2]^T, \\
  \mathbb{P}^{2 n +1 } & := [\tau_{\mathcal{M}(1)}, \tau_{\mathcal{M}(2)}, \cdots, \tau_{\mathcal{M}(n)}]^T + \delta [3, 2, \cdots, 2]^T , \\
  & \vdots \\
  \mathbb{P}^{n 2^B } & := [\tau_{\mathcal{M}(1)}, \tau_{\mathcal{M}(2)}, \cdots, \tau_{\mathcal{M}(n)}]^T + \delta [0, 0, \cdots, 0]^T .
  \end{array}
\end{aligned}
\end{equation}

As previously mentioned, for $\psi \in \mathbb{P}^l,l = 1,\cdots,n2^B$, the vector $\mathbf{\Omega}_{\text{opt}} (\psi)$ remains constant. The $\mathcal{M}(i)$-th element of $\mathbf{\Omega}_{\text{opt}} (\psi)$ is given as
\begin{equation}
  \Omega^l_{\mathcal{M}(i),\text{opt}} = \mathbb{P}^l(i) - \tau_{\mathcal{M}(i)}.
\end{equation}
Finally, we construct a search set $\mathbb{U}$ consisting of $n 2^B$ candidates
\begin{equation}
  \mathbb{U} = \left\{ \{ \Omega^1_{1,\text{opt}}, \ldots, \Omega^1_{n,\text{opt}} \}, \cdots, \{ \Omega^{n 2^B}_{1,\text{opt}}, \ldots, \Omega^{n 2^B}_{N,\text{opt}} \} \right\} .
\end{equation}

\begin{prop}
There exists a candidate set $\mathbb{U}$ with a cardinality of $n 2^B$ such that the discrete uni-modular constrained inner product maximization problem $\max_{ \mathbf{\Omega} \in \Delta^n } \lvert \langle \mathbf{v} , e^{ j \mathbf{\Omega} } \rangle \rvert $ is equivalent to 
\begin{equation}
  \max_{ \mathbf{\Omega} \in \mathbb{U} } \lvert \langle \mathbf{v} , e^{ j \mathbf{\Omega} } \rangle \rvert .
\end{equation}
\end{prop}

\begin{rem}
  The primary advantage of this encoding method is the significant reduction in the search space. The cardinality of $\mathbb{U}$ decreases from an exponential size to a polynomial size. The number of required searches decreases from $2^{n B}$ to $n 2^B$.
\end{rem}

The proposed DaS-based alternating inner product maximization for solving $(Q_1)$ and $(Q_2)$ is illustrated in Algorithm~\ref{alg1}.
%

\begin{algorithm}[htbp]
\caption{DaS-based alternating inner product maximization} 
\label{alg1} 
\begin{algorithmic}[1] 
\REQUIRE Complex matrix $\mathbf{A}$, the number of discrete levels $2^B$. 
\ENSURE Optimal discrete phase configurations $\mathbf{\Omega}_{\text{opt}}$. 
\STATE Choose an initial iteration point $\mathbf{\Omega}_0$, and set $k=0$ . 
\REPEAT
\STATE Calculate $\mathbf{z}_k$ according to \eqref{E:Arg_z}. 
\STATE Calculate $\lVert A e^{j \mathbf{\Omega}_k } \rVert_p$.
\STATE Determine $\mathbf{\Omega}_{k+1}$ by DaS method (\textbf{S1}-\textbf{S5}) based on $\mathbf{A}^H$ and $\mathbf{z}_k$ as in~\eqref{E:Arg_Omega}.
	\STATE \ \  \textbf{S1}: Set $\mathbf{v}=\mathbf{A}^H\mathbf{z}_k$, and calculate the magnitude $| v_i |$ \\ \ \ and angle $\tau_i$, for each element $v_i$ in $\mathbf{v}$.
    \STATE \ \ \textbf{S2}: Divide the complete angle $[0,2\pi)$ into $2^B$ parts\\ \ \ based on $\tau_i$ and construct the set $\mathcal{C}^i$ as~\eqref{E:CandidateSet} . 
    \STATE \ \ \textbf{S3}: Sort $\{ \tau_1, \ldots, \tau_n \}$ and find the mapping $\mathcal{M}$ as\\ \ \ in~\eqref{E:BijectiveMapping}.
    \STATE \ \ \textbf{S4} :Construct the set $\mathbb{U}$ consisting of $n 2^B$ candidates \\ \ \ according to the parts encoder in~\eqref{E:EncodingPartition}.
    \STATE \ \ \textbf{S5}: Find $\mathbf{\Omega}_{\text{opt}} = \arg \max_{ \mathbf{\Omega} \in \mathbb{U} } \ \lvert \langle \mathbf{v} , e^{ j \mathbf{\Omega} } \rangle \rvert$.
\STATE Calculate $\lVert A e^{j \mathbf{\Omega}_{k+1} } \rVert_p$.
\UNTIL $\left\lvert \lVert A e^{j \mathbf{\Omega}_{k+1} } \rVert_p - \lVert A e^{j \mathbf{\Omega}_k } \rVert_p \right\rvert \leq 10^{-10}$
\RETURN $\mathbf{\Omega}_{\text{opt}} = \mathbf{\Omega}_{k+1}$.
\end{algorithmic} 
\end{algorithm}

\subsection{Maximize the $\ell_\infty$-norm}

With the proposed DaS method, we can solve $(Q_\infty)$ efficiently. Utilizing the definition of the $\ell_\infty$-norm, we rewrite $(Q_\infty)$ as
\begin{equation}
  \max_{ \mathbf{\Omega} \in \Delta^n } \max_{i=1, \ldots, m} \lvert \langle \mathbf{A}_{i,:}^H, e^{j \mathbf{\Omega} } \rangle \rvert 
  = \max_{i=1, \ldots, m} \max_{ \mathbf{\Omega} \in \Delta^n } \lvert \langle \mathbf{A}_{i,:}^H, e^{j \mathbf{\Omega} } \rangle \rvert .
\end{equation}
Here, $\mathbf{A}_{i,:}$ represents the $i$-th row of the matrix $\mathbf{A}$. This equivalent formulation implies that to find the solution to $(Q_\infty)$, we only need to solve $m$ discrete uni-modular constrained inner product maximization problems, each of which can be efficiently solved using the proposed DaS method.

\section{Convergence and Lifting Capabilities}\label{Section6}

To evaluate the effectiveness of the proposed alternating inner product maximization approach for solving $\ell_p$-norm maximization problems, we conduct numerical assessments. It is easily seen that the problem $(Q_\infty)$ involves multiple independent inner product maximization problems, each of which can be effectively addressed by the DaS method without the need of iterations. We then focus on the performance of the alternating maximization framework for $(Q_1)$ and $(Q_2)$. 
  \begin{figure}[htbp]
    \centering
    \subfigure[]{
    \label{Fig:Lifting:1norm}
    \includegraphics[width=.75\columnwidth]{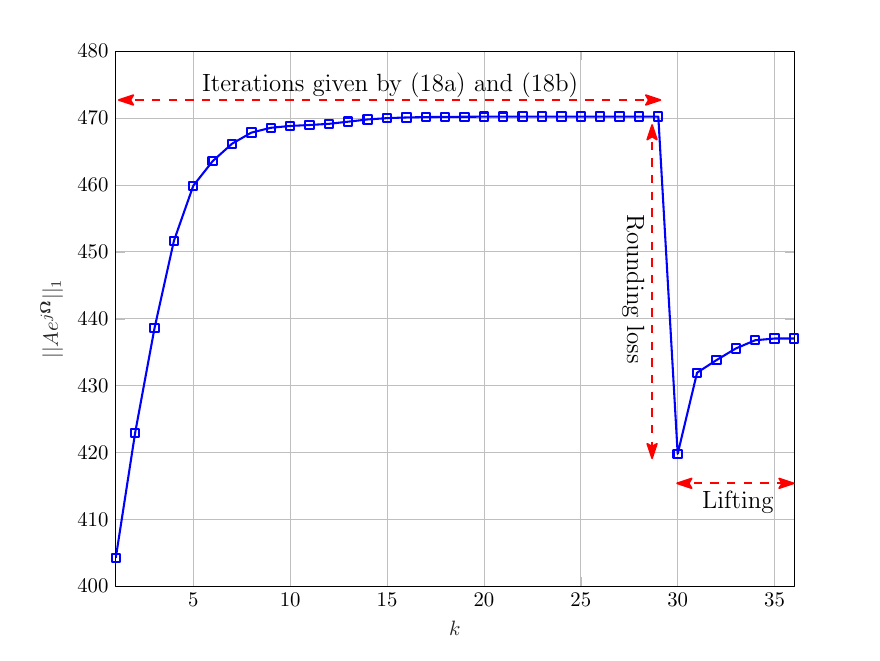}}
    \subfigure[]{
    \label{Fig:Lifting:2norm}
    \includegraphics[width=.75\columnwidth]{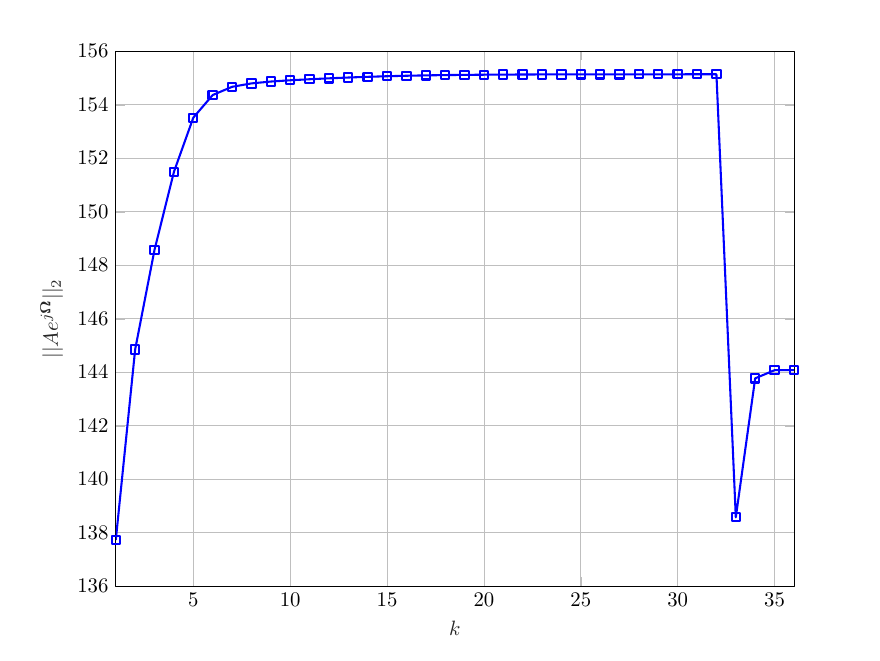}}
    \caption{Convergence behavior and lifting performance. We set $\mathbf{A} \in \mathbb{C}^{10 \times 100}$ and $B = 2$. Both discrete and continuous alternating inner product maximizations converge rapidly, requiring only a small number of iterations. The proposed lifting approach recovers the hard rounding loss. (a) and (b) illustrate the value of cost function for solving $(Q_1)$ and $(Q_2)$, respectively.}
    \label{Fig:Lifting}
  \end{figure}

In Fig.~\ref{Fig:Lifting}, we present the convergence behaviors of the proposed alternating inner product maximization schemes as outlined in (15) and (18). Figs.~\ref{Fig:Lifting:1norm} and \ref{Fig:Lifting:2norm} show that both discrete and continuous alternating inner product maximizations converge rapidly, requiring only a small number of iterations. Moreover, we observe there exists significant rounding loss, particularly with low-bit quantizations. However, after applying post-rounding lifting, there's a notable improvement in the cost-function performance. This unique lifting capability highlights the distinctiveness of the proposed alternating inner product maximization approach. 

  \begin{figure}[htbp]
    \centering
    \subfigure[]{
    \label{Fig:LiftingStat:1norm}
    \includegraphics[width=.75\columnwidth]{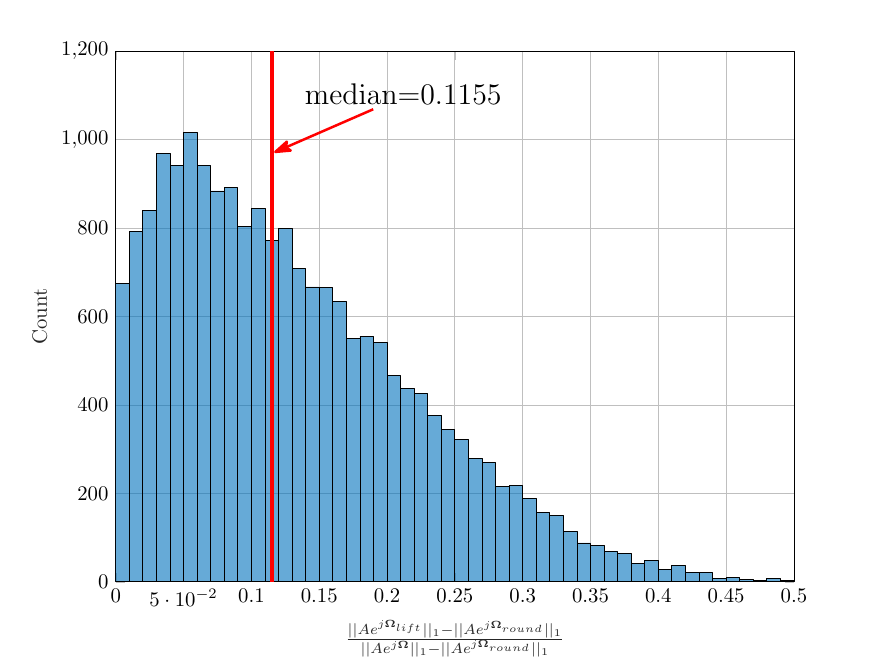}}
    \subfigure[]{
    \label{Fig:LiftingStat:2norm}
    \includegraphics[width=.75\columnwidth]{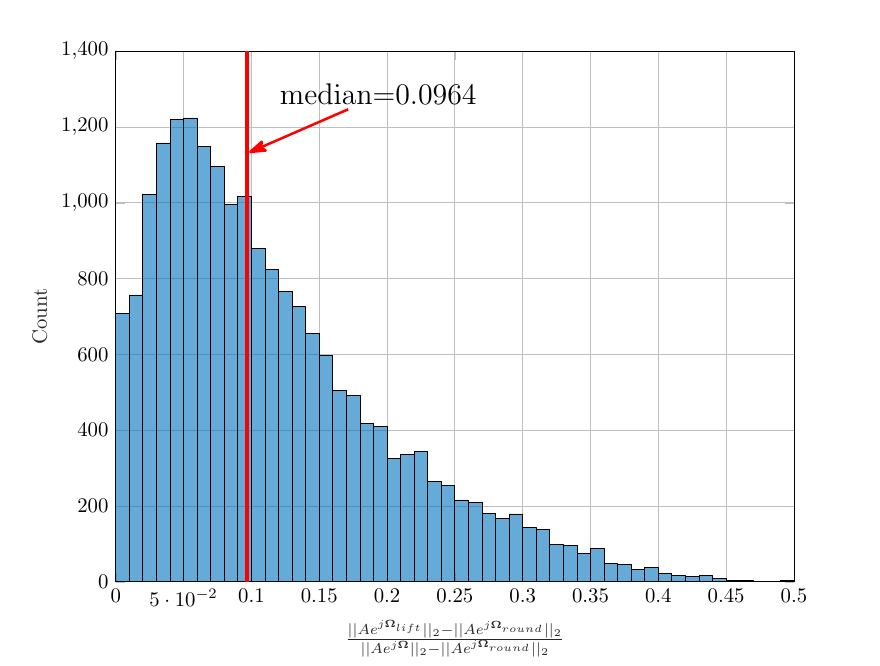}}
    \caption{Lift the solution obtained using Manopt. We set $\mathbf{A} \in \mathbb{C}^{10 \times 100}$ and $B = 1$, conducting 20,000 trials for $p=1, 2$. The distribution of the relative lifting gain is visualized, with the red line marking the median. (a) and (b) illustrate the results for solving $(Q_1)$ and $(Q_2)$, respectively.}
    \label{Fig:LiftingStat}
  \end{figure}

We now present the lifting performance for the discretization of the solution obtained using Manopt, a distinguished state-of-the-art technique.  For better illustration, we first define the relative lifting gain as the ratio between the lifting gain and the rounding loss, expressed as
\begin{equation}
  \frac{ \lVert A e^{j \mathbf{\Omega}_{\text{lifted}} } \rVert_p - \lVert A e^{j \mathbf{\Omega}_{\text{rounded}} } \rVert_p }{ \lVert A e^{j \mathbf{\Omega}_{\text{unrounded}} } \rVert_p - \lVert A e^{j \mathbf{\Omega}_{\text{rounded}} } \rVert_p } .
\end{equation}
We conduct 20,000 trials for $p=1, 2$ and record the relative lifting gain. Fig.~\ref{Fig:LiftingStat} illustrates the distribution of the relative lifting gain, with the red line indicating the median. We observe that there exist instances that the relative lifting gain is greater than 0.4, indicating that for these instances $40\%$ of the hard rounding loss is compensated for. The median value in Fig.~\ref{Fig:LiftingStat:1norm} indicates that, in solving $(Q_1)$, the proposed lifting approach recovers at least $11.55 \%$ of the hard rounding loss for half of the conducted trials. Similarly, in solving $(Q_2)$, at least $9.64 \%$ of the hard rounding loss could be recovered for half of the conducted trials.

\section{Numerical Comparisons of SNR Boosting in RIS beamforming}\label{Section7}

\begin{figure}[tbp]
  \centering
  \subfigure[]{
  \label{ComparisonM32B1}
  \includegraphics[width=0.75\linewidth]{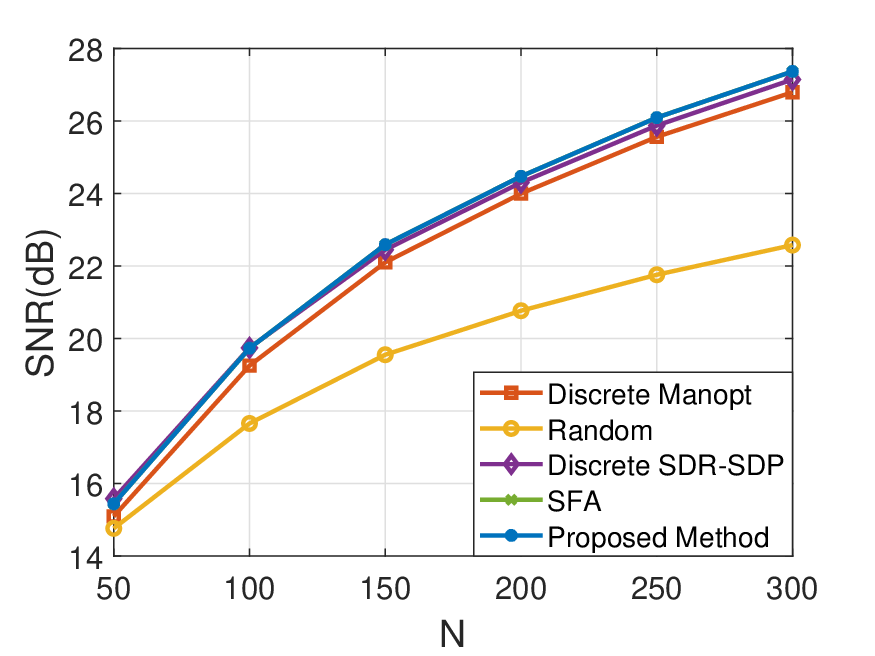}}
  \subfigure[]{
  \label{ComparisonM32B2}
  \includegraphics[width=0.75\linewidth]{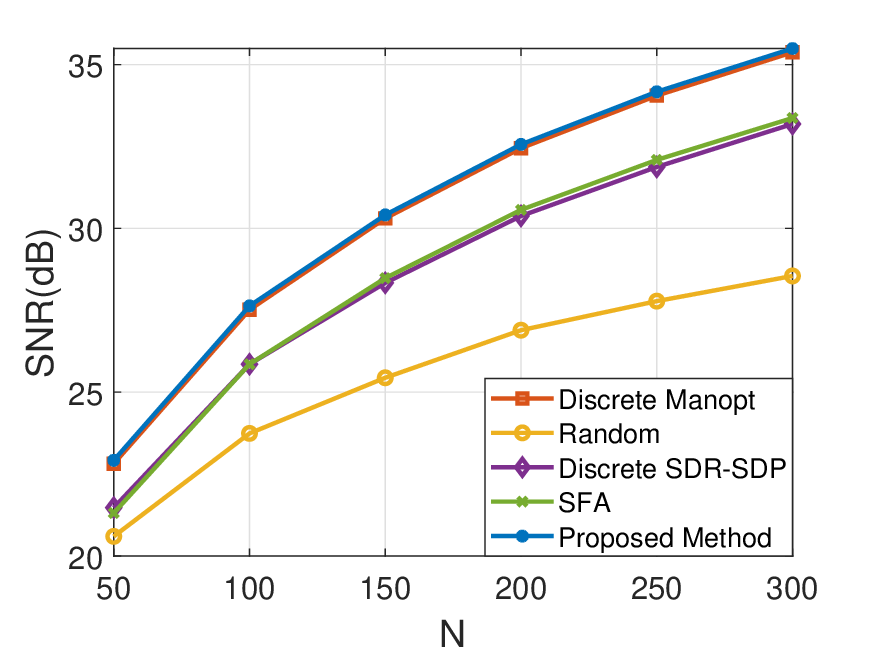}}
  \caption{A comparison of SNR performance of different methods as a function $N$ with the number of BS antennas $M=32$. (a) 1-bit phase quantization scheme. (b) 2-bit phase quantization scheme.}
  \label{ComparisonM32}
\end{figure}

In the context of SNR boosting through RISs, we focus on solving the $\ell_2$-norm maximization problems~\eqref{P4} and~\eqref{Optimization Problem}. We compare the performance of our proposed alternating inner product maximization approach with several state-of-the-art methods, including SDR-SDP~\cite{wu2019intelligent}, Manopt~\cite{yu2019miso}, successive refinement algorithm (SFA)~\cite{wu2019beamforming}, and approximation (APX)~\cite{zhang2022configuring}. SFA and APX are the methods that handle discrete phase configuration directly. APX is dedicated to single-output systems ($M=1$) and is claimed to approximate the global optimum. Additionally, we include a search among $\mathbf{10^5}$ random configurations as a benchmark, denoted as Random. For clarity, we refer to discrete SDR-SDP and discrete Manopt as the hard rounding of their continuous counterparts.

We conduct a comprehensive evaluation of the DaS-based inner product maximization algorithm's performance under various parameter settings. The power gain is recorded for each of the 100 trials with a fixed value of $N$. The channels are assumed to follow independent and identically distributed (i.i.d.) Gaussian distributions with zero mean and variance $\sigma^2$. The noise is Gaussian with variance 1.



Initially, we investigate the discrete beamforming in the NLOS channels, considering an increased number of transmitting antennas $M$. The results depicted in Fig.~\ref{ComparisonM32} illustrate the superior performance of the proposed DaS-based alternating inner product maximization approach compared to other competing methods in terms of SNR. These plots highlight the robustness of our approach in achieving optimal solutions. Notably, for $B=1$ and $B=2$, SFA emerges as the second-best approach. However, as $B$ increases to 2, the performance of SFA deteriorates significantly, dropping more than 2 dB below the DaS-based alternating inner product maximization approach.

In Fig.~\ref{CDFM32B2}, we proceed to investigate the scenario with $M=32$ while keeping $N$ fixed at 200. Across 1000 trials, we observe the power gain and plot the cumulative distribution functions of SNR. The results consistently demonstrate that the DaS-based alternating inner product approach outperforms other methods, achieving an average SNR gain of 2 dB with the quantization level $B=2$ compared to SFA and the hard-rounding of SDR-SDP.
\begin{figure}[tbp]
\centering
\includegraphics[width=0.8\columnwidth]{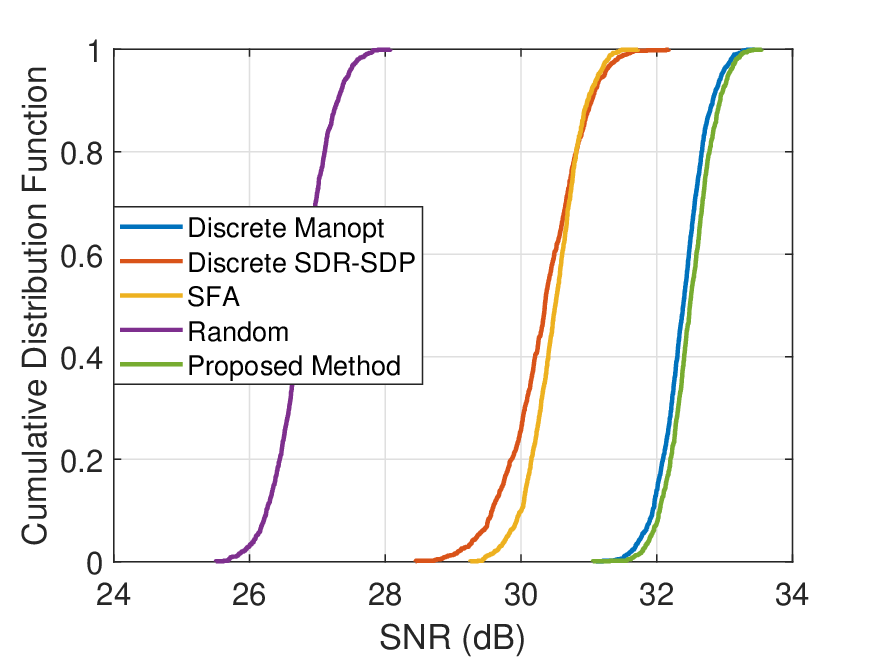}
\caption{Cumulative distribution of SNR versus different methods under the quantization level $B=2$ and number of units $N=200$, and The number of BS antennas $M=32$.}
\label{CDFM32B2}
\end{figure}
\vspace{-0.2cm}

%

\begin{figure}[htbp]
  \centering
  \subfigure[]{
  \label{fig: Compa}
  \includegraphics[width=.75\columnwidth]{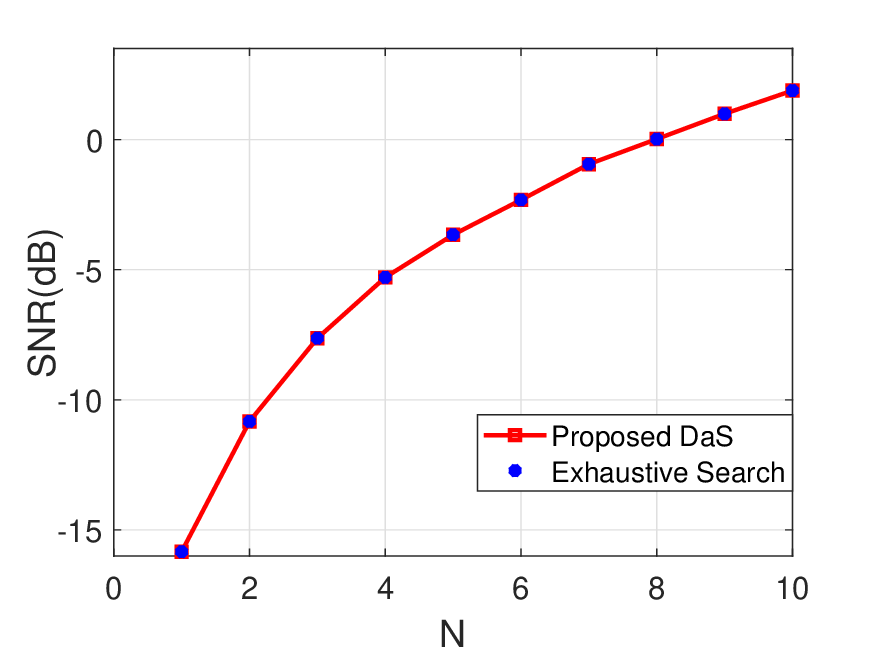}}
  \subfigure[]{
  \label{fig: Compa3Method}
  \includegraphics[width=.75\columnwidth]{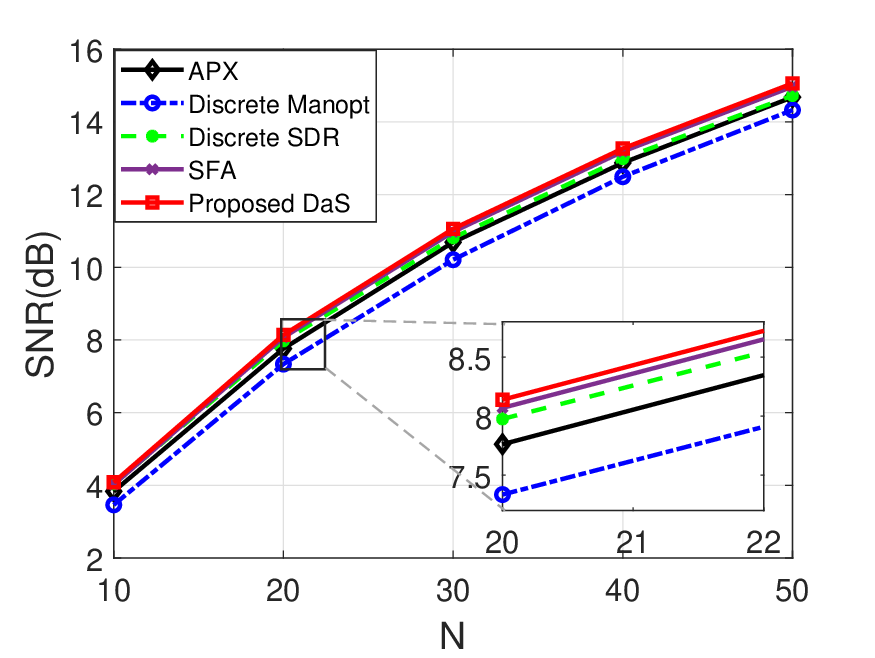}}
  \caption{A comparison of SNR performance as a function $N$ with 1-bit phase quantization scheme. The number of BS antennas $M=1$. (a) The proposed DaS algorithm achieves comparable performance to that of the exhaustive search method. (b) The SNR performance comparison of different methods.}
  \label{Fig1}
\end{figure}

As a special case, when the BS consists of only one antenna, i.e., $M=1$, the optimal beamforming problem is reduced into a discrete inner product maximization problem. This can be efficiently addressed using our proposed DaS search method. The plots in Fig.~\ref{fig: Compa} demonstrate that in such scenarios, the DaS achieves SNR results identical to those obtained through exhaustive search. In Fig.~\ref{fig: Compa3Method}, the comparison results show that when using the 1-bit quantization scheme, the proposed method outperforms other competitors. 



\subsection{4-bit Quantization is Adequate}

In Fig.~\ref{Comparison-Continuous}, we present the results of continuous phase configurations (obtained using Manopt) depicted by red curves, aiming to offer an analysis regarding different quantization schemes. We consider the case regarding the number of BS antennas $M=16$. Notably, the 1-bit discrete configurations exhibit a loss of approximately 3 dB compared to the continuous phase configurations (Manopt) in this case. However, as the quantization resolution increases, the loss in received signal power diminishes. Remarkably, when utilizing 4-bit quantization, the proposed methods achieve SNR gains comparable to continuous configurations, with an average loss of less than  \textbf{0.02} dB. 




\begin{figure}[htbp]
\centering
\includegraphics[width=0.8\columnwidth]{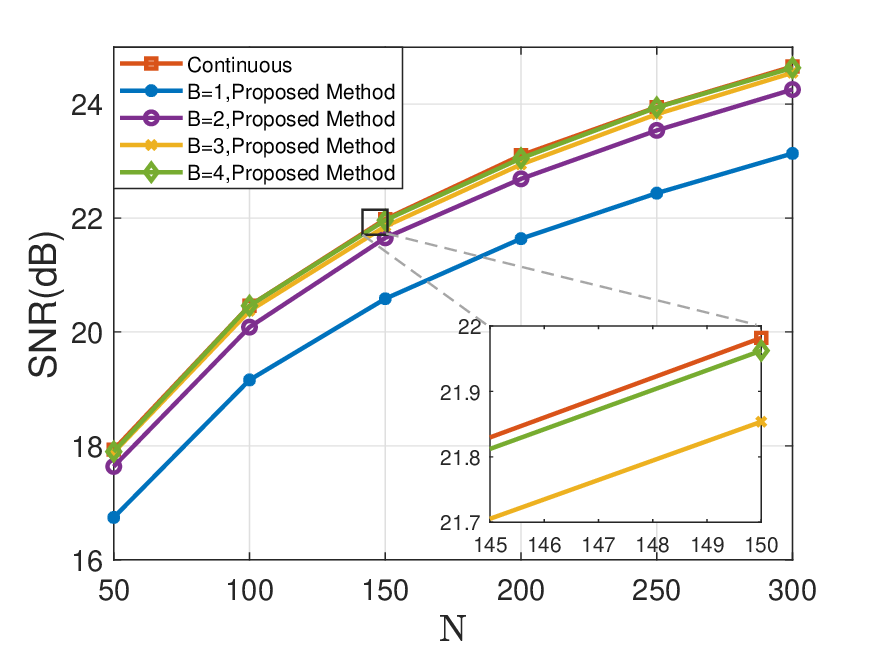}
\caption{The SNR performance distance between different quantization levels and the continuous phase configuration in NLoS scenarios, $N=200$ and $M=16$.}
\label{Comparison-Continuous}
\end{figure}

\subsection{Overall Execution-time Comparison}

To assess the effectiveness of the proposed approach, we conduct tests to evaluate the execution-time as a function of the number of reflecting units in the setting of $M=32$. For each value of $N$, we conduct 100 trials and record the overall execution-time, as presented in Table~\ref{tab: time} and Fig.~\ref{fig:time}. The results demonstrate that the execution-time of discrete SDR-SDP increases substantially as $N$ increases. In fact, for $N = 1000$, it fails to produce a solution even after running for many hours. Conversely, the proposed DaS-based alternating inner product maximization method is much more efficient, when $N = 1000$, it takes only 0.7409 s on average to find the optimal discrete phase configurations. The SFA algorithm is the second fastest. The superior execution-time performance of DaS-A algorithm makes it particularly well-suited for large-scale practical implementations.

\begin{table*}[t]
\caption{Overall execution-time comparison for 100 channel realizations of 1-bit RIS, $M=32$} 
\centering
\setlength{\tabcolsep}{2.4mm}
\begin{tabular}{cccccccc}
\toprule
Methods & $N=10$  &  $N = 50$ & $N = 100$ & $N = 200$ &$N = 500$ &$N = 1000$  \\
\midrule
Discrete SDR-SDP  &  $208.06$ s &$241.01$ s  &$411.57$ s  &$1536.24$ s  & $22767.74$ s & -                         \\
Discrete Manopt  &  $3.27$ s   &$14.02$ s    & $23.76$ s  &$58.57$ s   & $515.25$ s  & $1164.71$ s               \\
SFA       &    $0.04$ s      &$0.15$ s    &$0.44$ s    &$2.15$ s    &$19.03$ s    & $134.85$ s      \\
Proposed Method   &  $0.03$ s   &$0.09$ s    & $0.22$ s   &$0.72$ s    & $11.44$ s   & $74.09$ s      \\
\bottomrule 
\vspace{-0.5cm}
\label{tab: time}
\end{tabular}
\end{table*}

\begin{figure}[htbp]
\centering
\includegraphics[width=0.8\columnwidth]{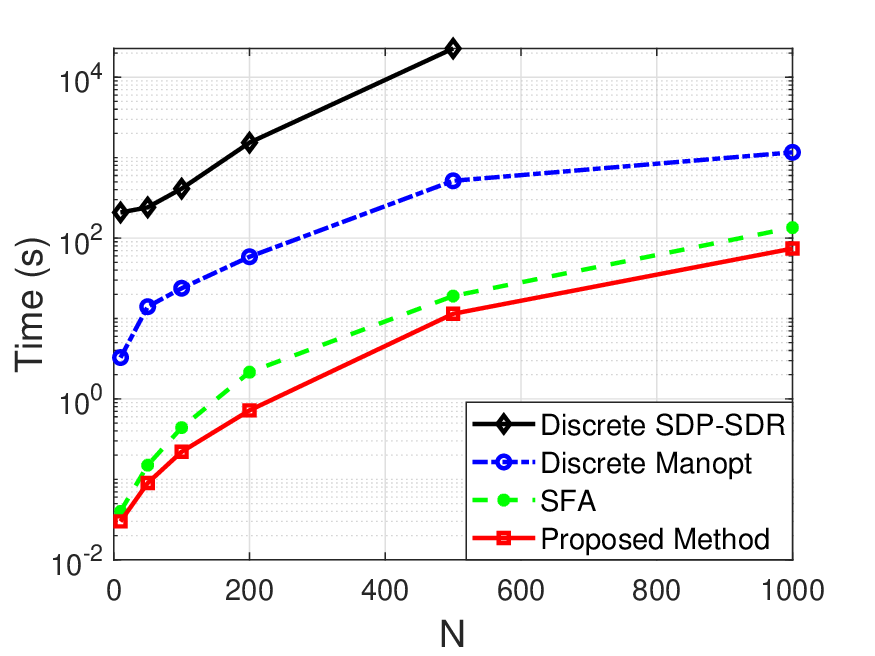}
\caption{The overall execution-time of $100$ trials for each $N$ in NLOS scenarios, with a quantization level of $B=1$, and the number of BS antennas is $M=32$.}
\label{fig:time}
\end{figure}

\section{RIS-Aided Experimental Field Trials}\label{Section8}
To evaluate the performance of the proposed algorithms in real-world scenarios, we conduct a series of experiments within a practical RIS-aided communication system. One of the key challenges in RIS-aided communications is the lack of dedicated signal processing capabilities in passive reflecting units, which makes traditional channel estimation methods impractical. However, in the context of geometrical optics models~\cite{mi2023towards}, the phase differences between various paths are calculated directly based on units' locations and configurations. This feature enables beamforming without the need for explicit channel estimation processes, making it more efficient and practical in real-world applications. To ensure a fair and accurate comparison of the performance of various methods, the phase configurations corresponding to different methods are first derived using geometrical optics models. Subsequently, these configurations are translated into control codebooks and seamlessly implemented on the controller of the RISs, ensuring their readiness for practical testing.

\begin{figure}[htbp]
\centering
\includegraphics[width=0.9\linewidth]{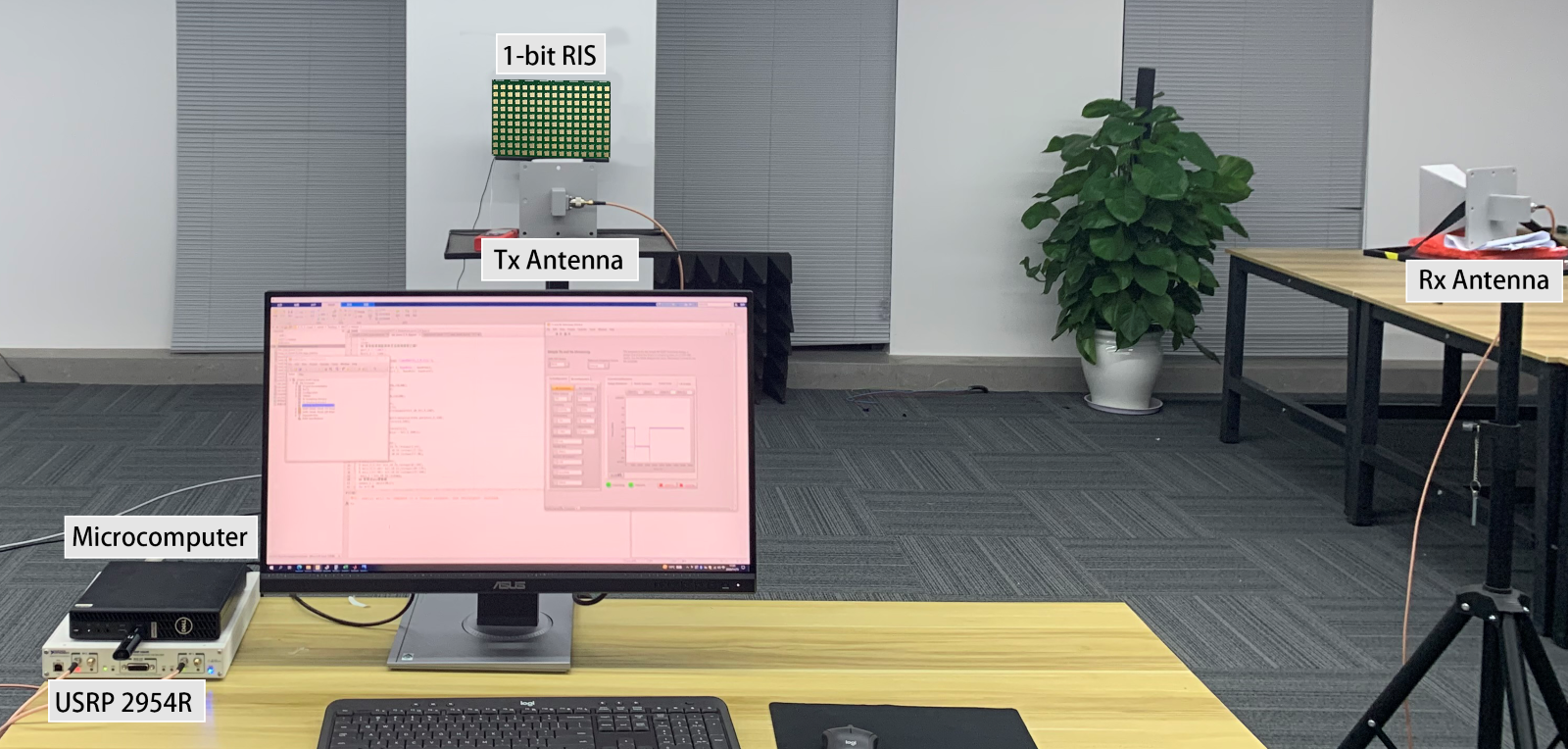}
\caption{The prototype of 1-bit RIS-aided wireless communication system.}
\label{system}
\end{figure}

We carry out all experiments utilizing a 1-bit RIS-aided communication system, as depicted in Figs.~\ref{system} and \ref{chamber}. The signal is generated and modulated using the USRP 2954R, transmitted through a horn antenna, reflected by the RIS, and finally received by another horn. To ensure the reliability of the evaluation, we calculate the average received signal power based on 8912 samples. The experiments are designed to encompass two different frequency types of RISs, providing a comprehensive evaluation and validation of the algorithm's performance under diverse scenarios.

\begin{figure}[tbp]
\centering
\includegraphics[width=0.9\columnwidth]{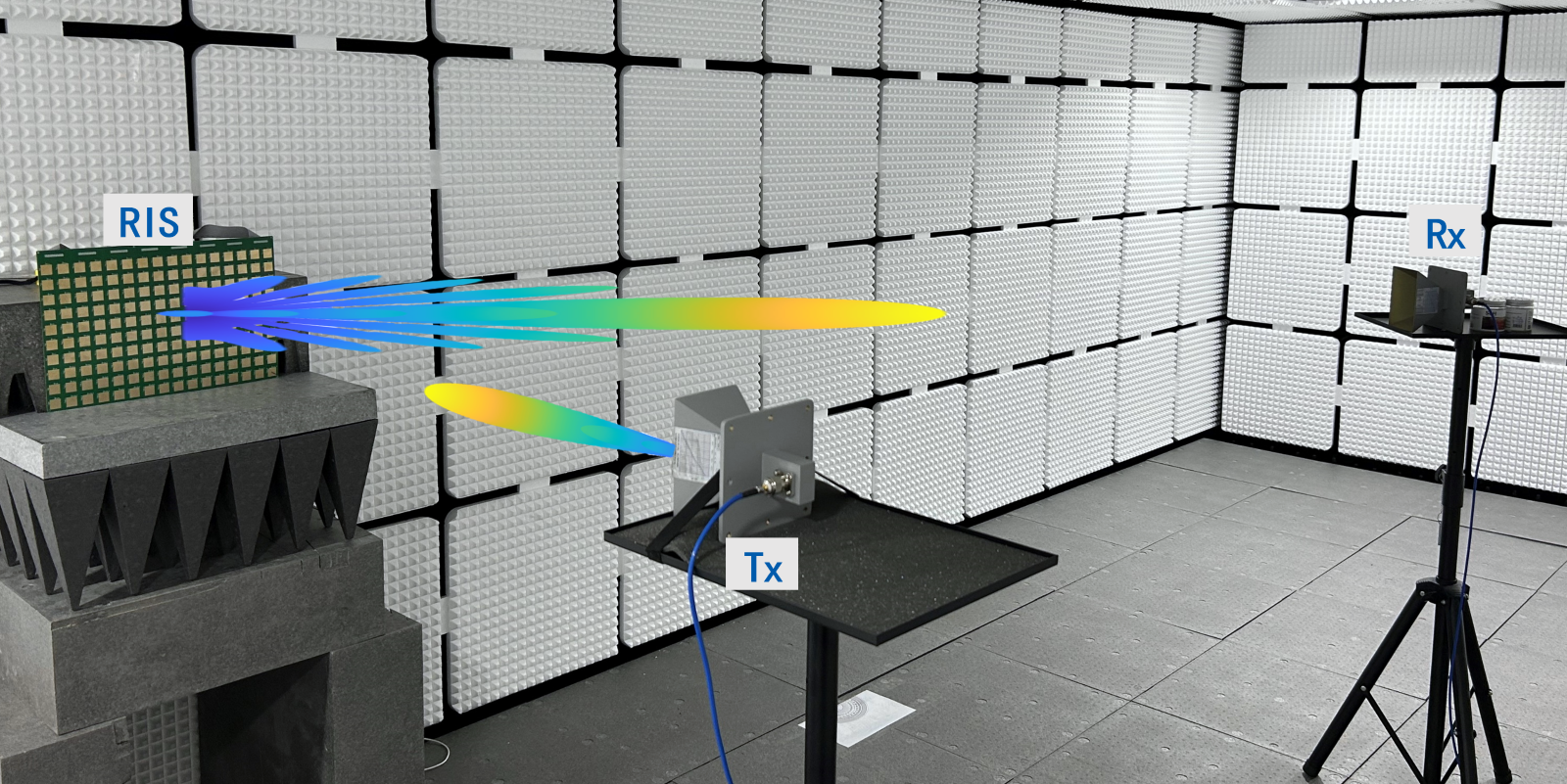}
\caption{Experiment conducted inside the anechoic chamber while using a 1-bit RIS operating at a frequency of 4.85 GHz.}
\label{chamber}
\end{figure}

\begin{figure}[tbp]
\centering
\includegraphics[width=0.9\linewidth]{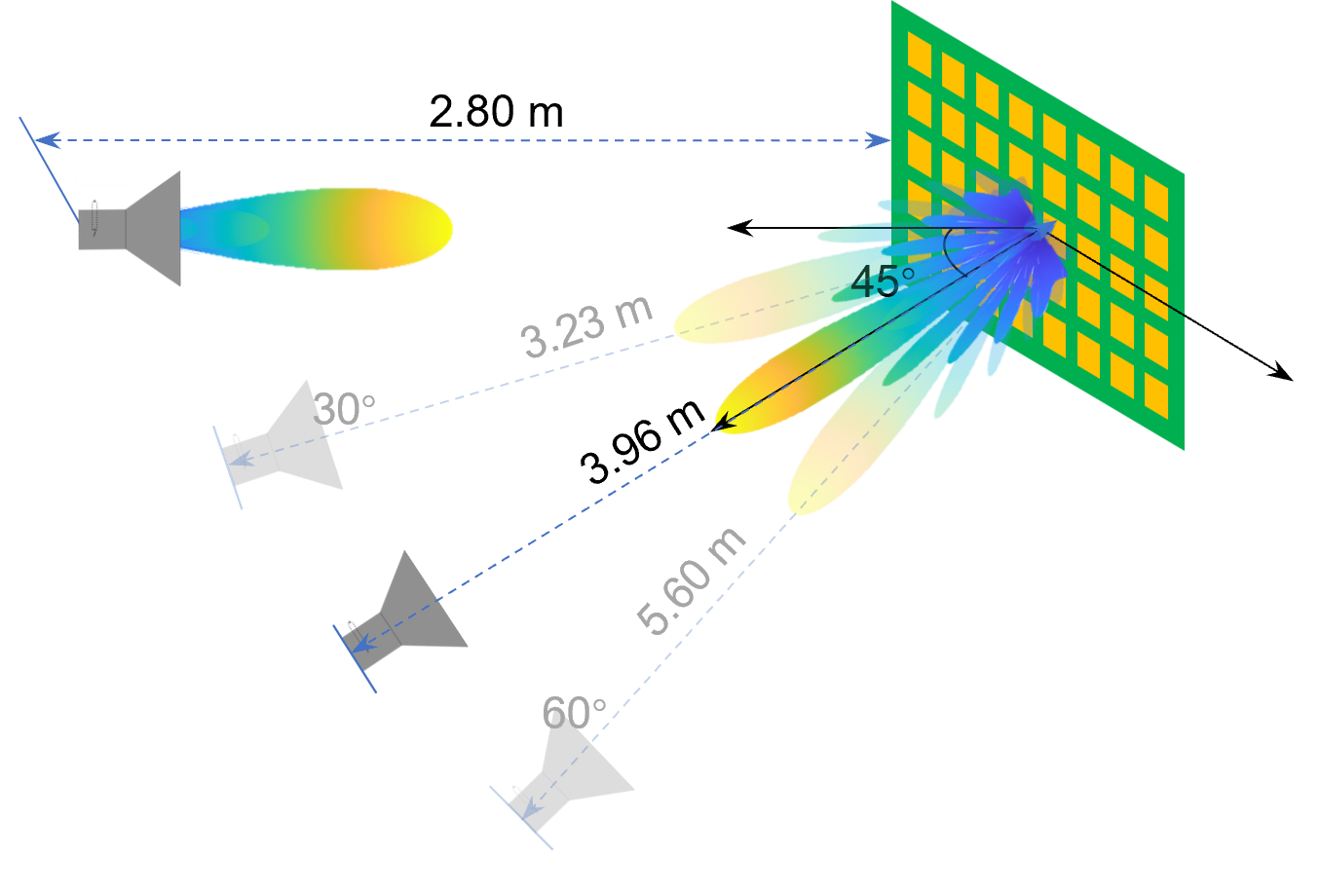}
\caption{Experiment setups while using the RIS operating at a frequency of 4.85 GHz. All devices are positioned at the same height to ensure that the azimuth angle $\phi$ remains equal to $0^\circ$ the Tx antenna is located at an elevation angle of $\theta^{\text{Arr}}=0^{\circ}$ with a distance of 2.8 meters to the center of RIS, the Rx antenna is positioned at direction angles of $\theta^{\text{Arr}}= 30^{\circ}, 45^{\circ}$, and $60^{\circ}$.}
\label{Experiment Test}
\end{figure}

\subsubsection{Anechoic Chamber Environment Testing}
We first conduct an experiment inside an anechoic chamber. The experiment aims to compare the received power before and after beamforming using the proposed algorithm. The 1-bit RIS utilized in this test operates at a frequency of 4.85 GHz, as illustrated in Fig.~\ref{chamber}. It comprises $10\times16$ reflecting units, with each unit spaced 0.027 meters apart. To simplify the experiments, all devices are positioned at the same height to ensure that the azimuth angle $\phi$ remains equal to $0^\circ$.

\begin{figure}[htbp]
  \centering
  \subfigure[]{
  \label{30}
  \includegraphics[width=.31\columnwidth]{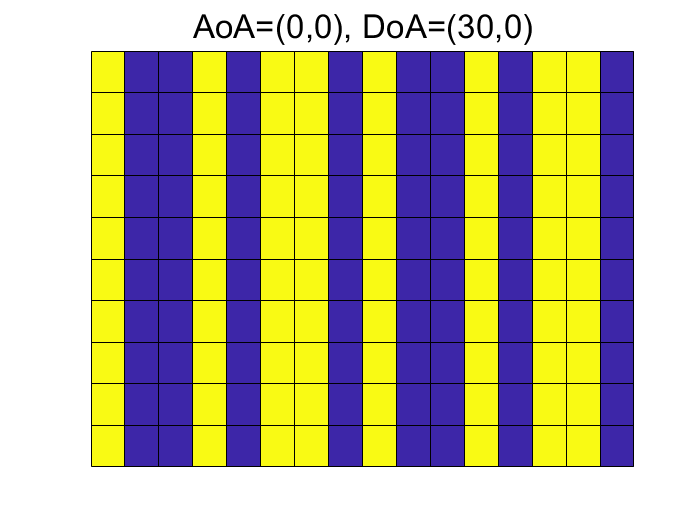}}
  \subfigure[]{
  \label{45}
  \includegraphics[width=.31\columnwidth]{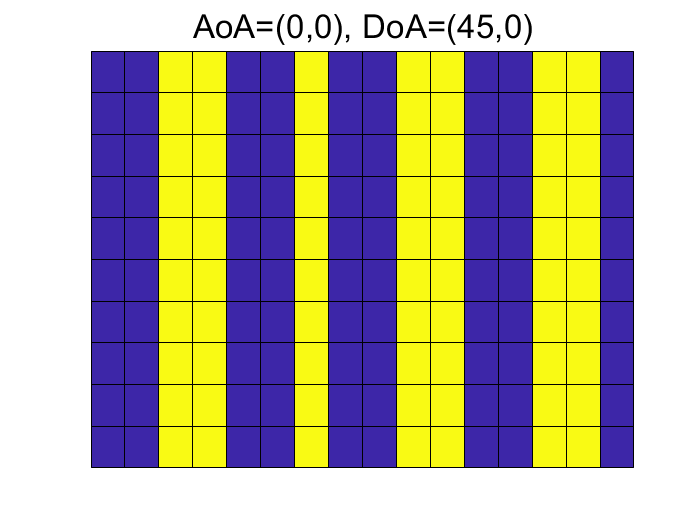}}
  \subfigure[]{
  \label{60}
  \includegraphics[width=.31\columnwidth]{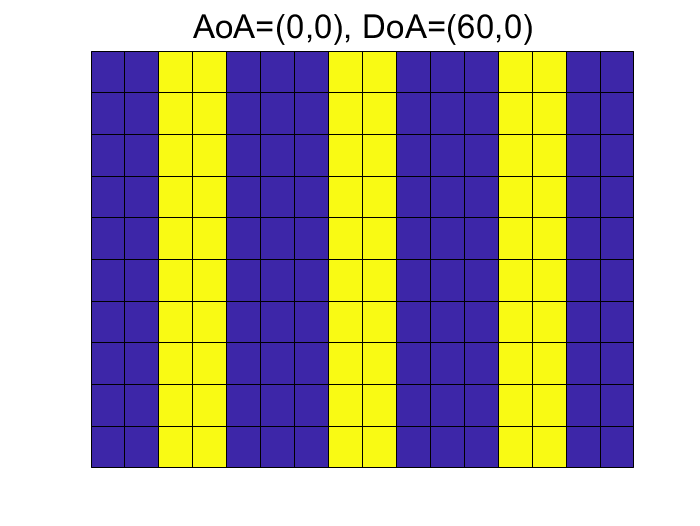}}
  \caption{RIS codebooks related to the optimization phase configurations by DaS when using the RIS operating at a frequency of 4.85 GHz. The Tx is located at the elevation angles $\theta^{\text{Arr}}=0^{\circ}$ with a distance of 2.8 meters to the center of RIS, the Rx is positioned at (a) $\theta^{\text{Arr}} = 30^{\circ}$, 3.23 meters. $\theta^{\text{Arr}} = 45^{\circ}$, 3.96 meters.  $\theta^{\text{Arr}} =60^{\circ}$, 5.6 meters.}
  \label{codebook}
\end{figure}

For each of the three experimental setups, we perform the proposed DaS method to obtain the optimal discrete phase configurations based on the geometrical optics model~\cite{mi2023towards}. These configurations are then transformed into control codebooks. As shown in Fig.\ref{codebook}. The bright yellow color represents the corresponding reflecting units in the configuration of $\Omega_i = 0$, with the controlling diodes in the ON state. Conversely, the blue color corresponds to the configuration of $\Omega_i = \pi$, indicating the OFF state. For clarity, the configuration with all diodes in the OFF state is denoted as "not-configured."

\begin{figure}[htbp]
  \centering
  \subfigure[]{
  \label{30degree}
  \includegraphics[width=.68\columnwidth]{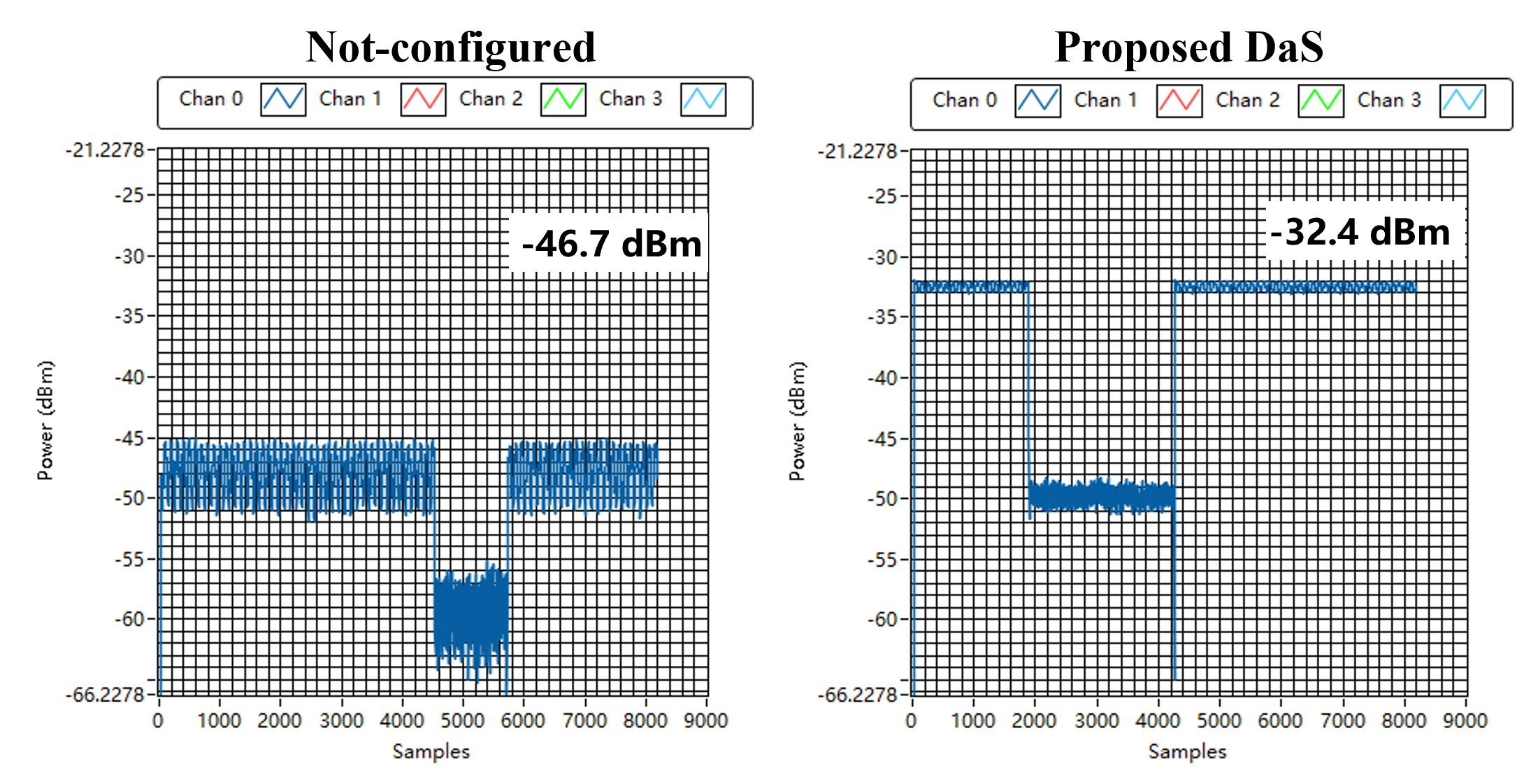}}
  \subfigure[]{
  \label{45degree}
  \includegraphics[width=.68\columnwidth]{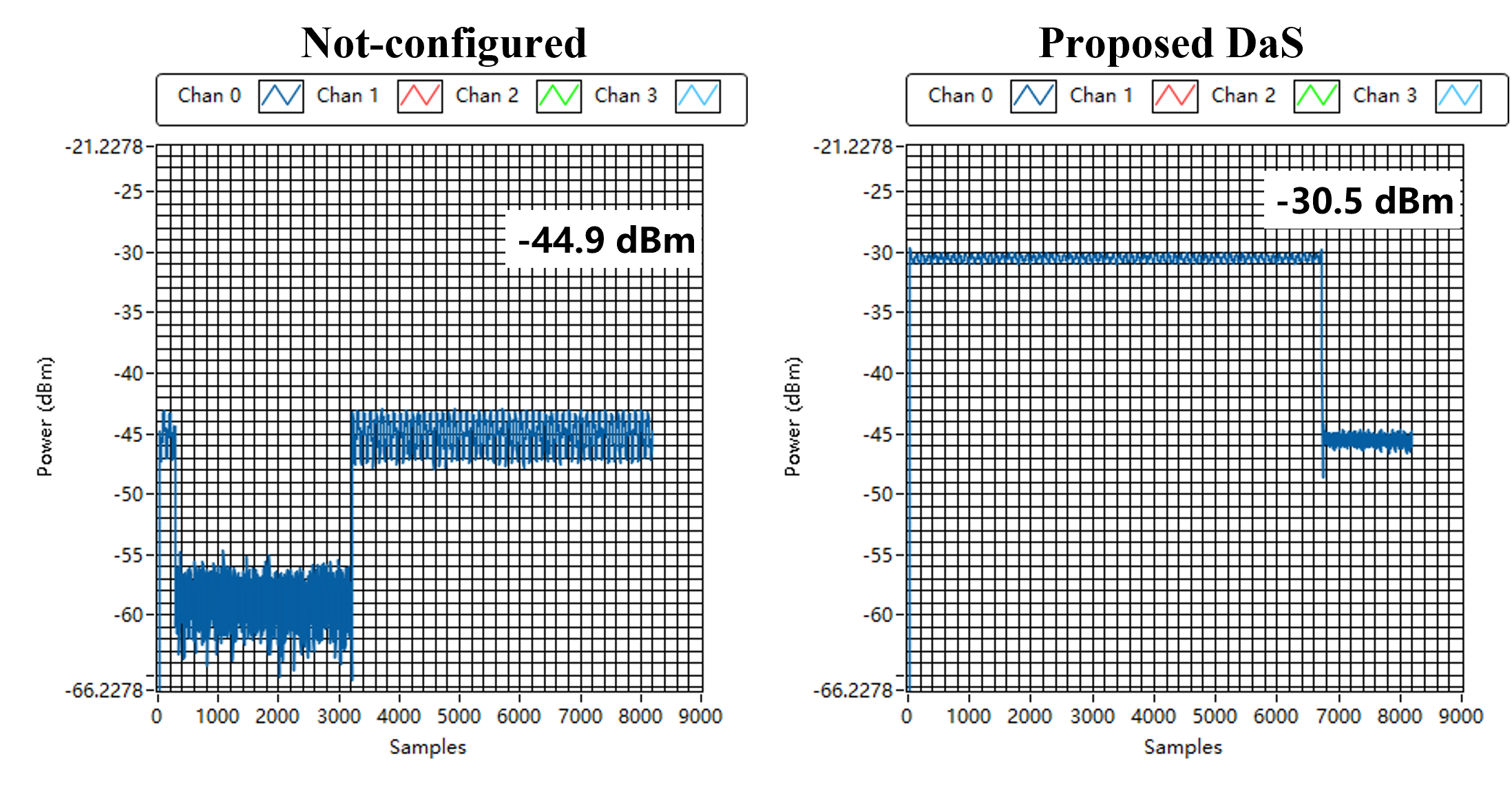}}
  \subfigure[]{
  \label{60degree}
  \includegraphics[width=.68\columnwidth]{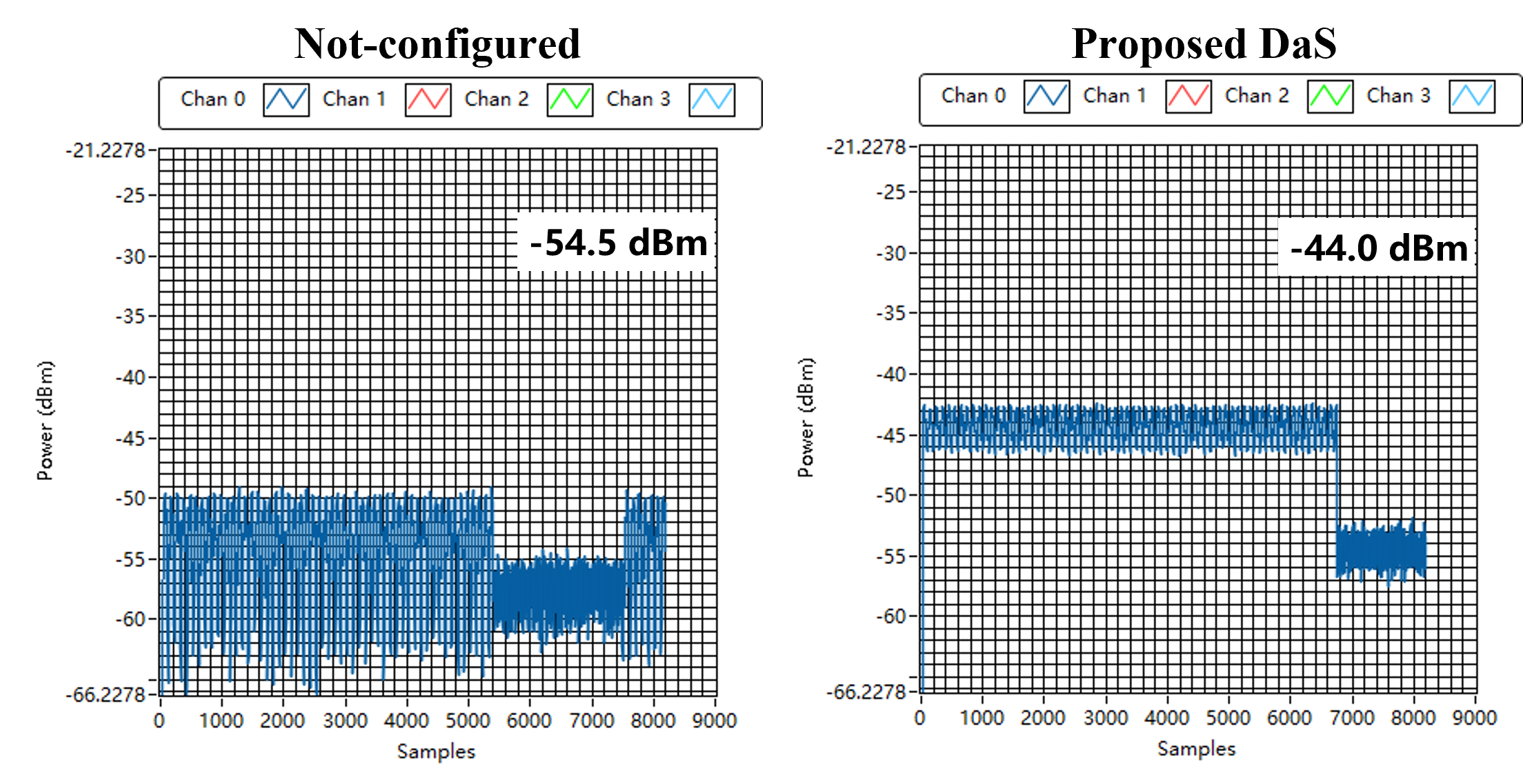}}
  \caption{The received signal power measurements at different locations during the experiment involving the RIS operating at a frequency of 4.85 GHz. The Tx keeps the distance at 2.8 meters in front of the RIS, with both elevation and azimuth angles set to $0^{\circ}$. The Rx is placed at (a) $\theta^{\text{Arr}}=30^{\circ}$, 3.23 meters. (b) $\theta^{\text{Arr}}=45^{\circ}$, 3.96 meters. (c) $\theta^{\text{Arr}}=60^{\circ}$, 5.6 meters.}
  \label{Repower}
\end{figure}

The experimental results in Fig.~\ref{Repower} highlight the remarkable performance of the optimized configurations derived through the proposed algorithm, resulting in substantial enhancements in signal strength. Across all of the experimental setups, the configurations designed by DaS achieve an average gain of 13.1 dB compared to no elaborate configurations.
\begin{figure}[htbp] 
  \centering
  \subfigure[]{
  \label{rand45}
  \includegraphics[width=.31\columnwidth]{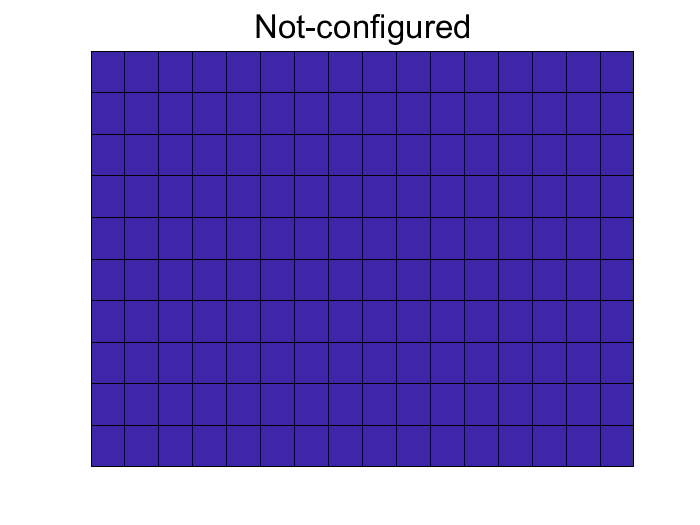}}
 \subfigure[]{
  \label{zero45}
  \includegraphics[width=.31\columnwidth]{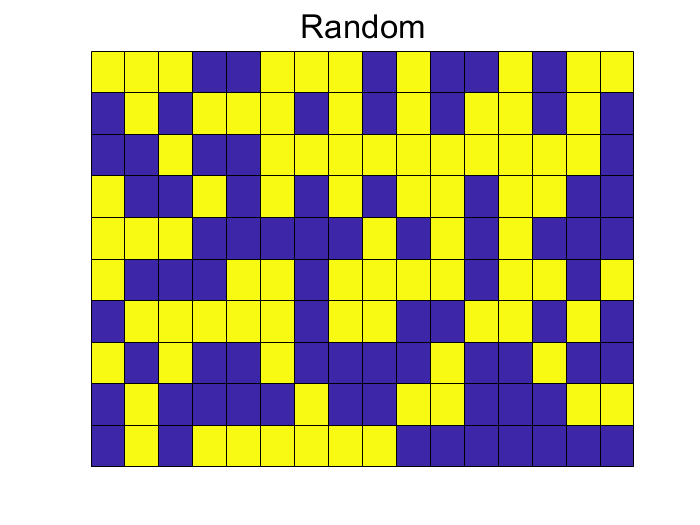}}
 \subfigure[]{
  \label{sdr45}
  \includegraphics[width=.31\columnwidth]{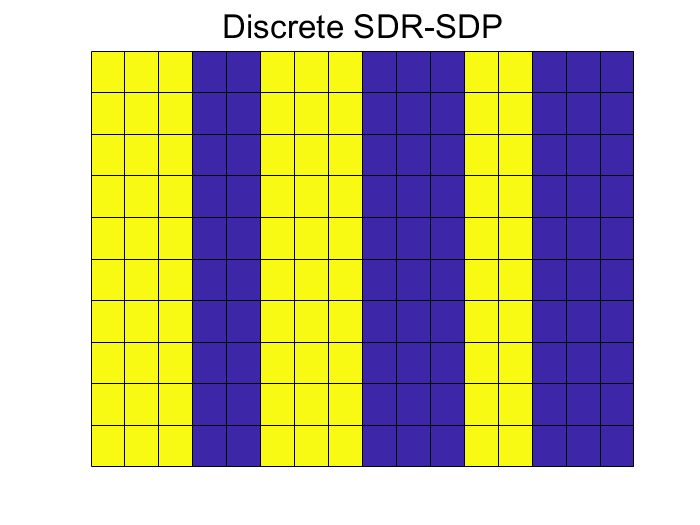}}
 \subfigure[]{
  \label{man45}
  \includegraphics[width=.31\columnwidth]{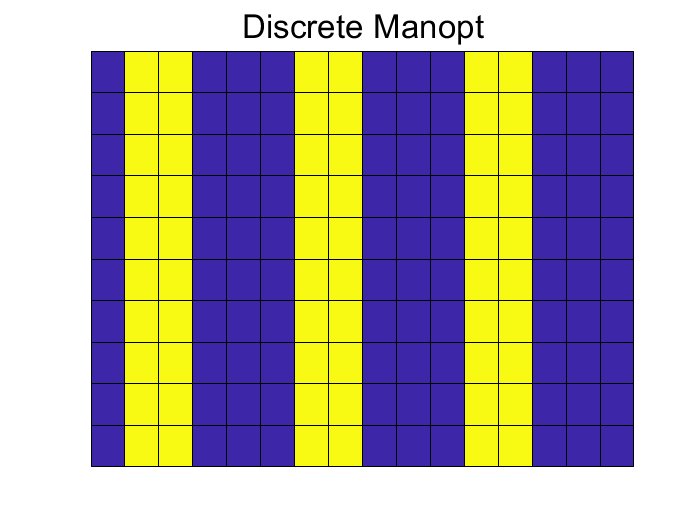}}
  \subfigure[]{
  \label{my45}
  \includegraphics[width=.31\columnwidth]{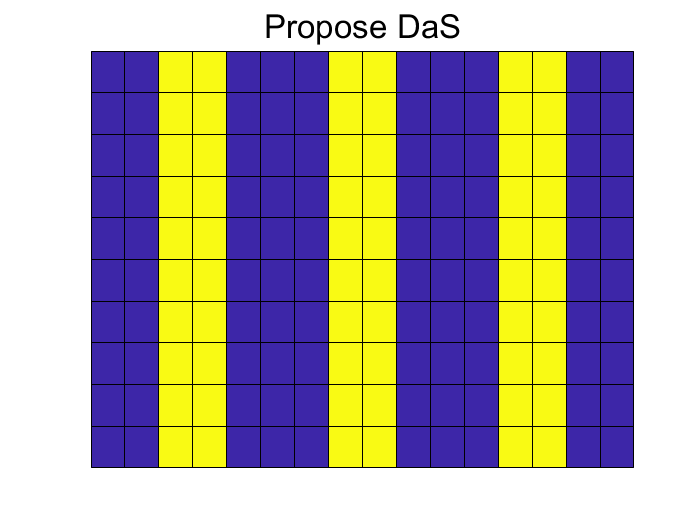}}
  \subfigure[]{
  \label{apx45}
  \includegraphics[width=.31\columnwidth]{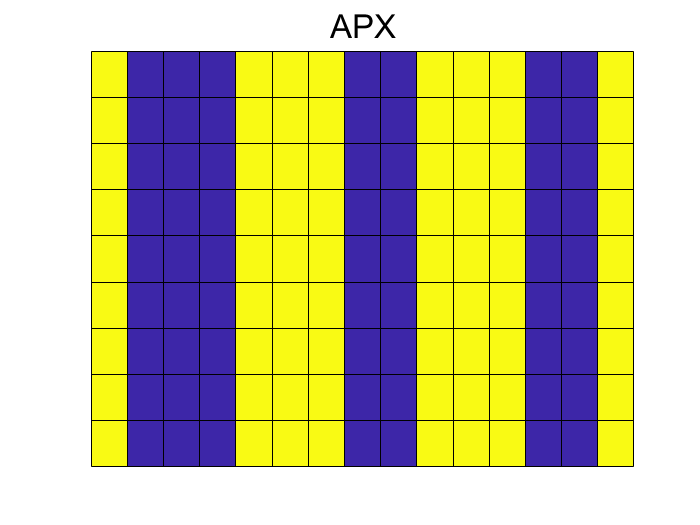}}
  \caption{The control codebooks correspond to the phase configurations obtained by various approaches at the operating frequency of 5.8 GHz. (a) Random. (b) Not-configured. (c) Discrete SDR-SDP. (d) Discrete Manopt. (e) Proposed DaS. (f) APX.}
\label{codebooks}
\end{figure}

\begin{table*}[t]
\centering
\begin{small}
\caption{Practical power gains over different approaches in 1-bit RIS operating at a frequency of 5.8 GHz} 
\centering
\setlength{\tabcolsep}{1.3mm}
\begin{tabular}{ccccccccc}
\toprule
Tx, Rx Position/Methods                         & Not-configured        &  Random    & Discrete SDR-SDP & Discrete Manop   & APX & Proposed DaS   \\
\midrule
Tx: 2.1 m, Rx: ($30^{\circ},0^{\circ})$ & $-47.29$ dBm  & $-44.81$ dBm  & $-39.83$ dBm    & $-40.59$ dBm     & $-40.54$ dBm  & $\mathbf{-32.79}$ dBm  \\
Tx: 2.1 m, Rx:  ($45^{\circ},0^{\circ})$ & $-49.42$ dBm  & $-44.68$ dBm  & $-33.09$ dBm    & $-33.21$ dBm     & $-32.52$ dBm  & $\mathbf{-31.63}$ dBm  \\
Tx: 2.1 m, Rx: ($60^{\circ},0^{\circ})$ & $-44.41$ dBm  & $-47.00$ dBm  & $-44.94$ dBm    & $-39.77$ dBm     & $-41.04$ dBm  & $\mathbf{-37.47}$ dBm   \\ 
Tx: 4.0 m, Rx: ($30^{\circ},0^{\circ})$ & $-40.00$ dBm  & $-45.14$ dBm  & $-43.74$ dBm    & $-42.70$ dBm     & $-40.78$ dBm  & $\mathbf{-39.91}$ dBm   \\
Tx: 4.0 m, Rx: ($45^{\circ},0^{\circ})$ & $-46.72$ dBm  & $-41.32$ dBm  & $-44.00$ dBm    & $-52.91$ dBm     & $-46.85$ dBm  & $\mathbf{-39.80}$ dBm    \\
Tx: 4.0 m, Rx: ($60^{\circ},0^{\circ})$ & $-52.21$ dBm  & $-48.11$ dBm  & $-46.58$ dBm    & $-44.32$ dBm     & $-46.50$ dBm  & $\mathbf{-43.75}$ dBm \\
\bottomrule
\label{Table-Performances}
\end{tabular}
\end{small}
\end{table*}

\subsubsection{Open Office Area Testing}
The primary objective of our second experiment is to compare the actual signal power gains achieved by various algorithms in an open office area. For this purpose, we consider the proposed DaS method, along with the discrete SDR-SDP and discrete Manopt methods. Additionally, we include two benchmark configurations: random phase configurations and not-configured phase configurations. The codebooks are visualized in Fig.~\ref{codebooks}. 

In this experiment, a 1-bit RIS operating at a frequency of 5.8 GHz is employed. The transmitting antenna is precisely placed 2.1 meters in front of the RIS board, with elevation angles of $\theta^{\text{Arr}}=0^{\circ}$. Conversely, the receiving antenna is set at a distance of 2.97 meters from the RIS, with an angle of $\theta^{\text{Dep}}=45^{\circ}$. Similarly, all devices are fixed at the same height.

The measurement results presented in Fig.~\ref{Result2.1-45} clearly demonstrate the effectiveness of the proposed DaS method, surpassing other competing methods by approximately 1 dB in terms of received signal power gain. Notably, when compared to no elaborate phase configurations, the proposed DaS method achieves a remarkable power gain of up to 18 dB.

We further conduct tests with the receiver positioned at direction angles $\theta^{\text{Arr}}=30^{\circ}$ and $\theta^{\text{Arr}}=60^{\circ}$, at distances of 2.42 meters and 4.20 meters from the RIS, respectively. Moreover, for comprehensiveness, we extend the transmitter-RIS distance to 4.0 meters and repeat the tests. The results are summarized in TABLE~\ref{Table-Performances}. It is evident from the table that the proposed approach achieves an increase in received signal power, outperforming other competing methods.

\begin{figure}[htbp]
\centering
\includegraphics[width=.98\columnwidth]{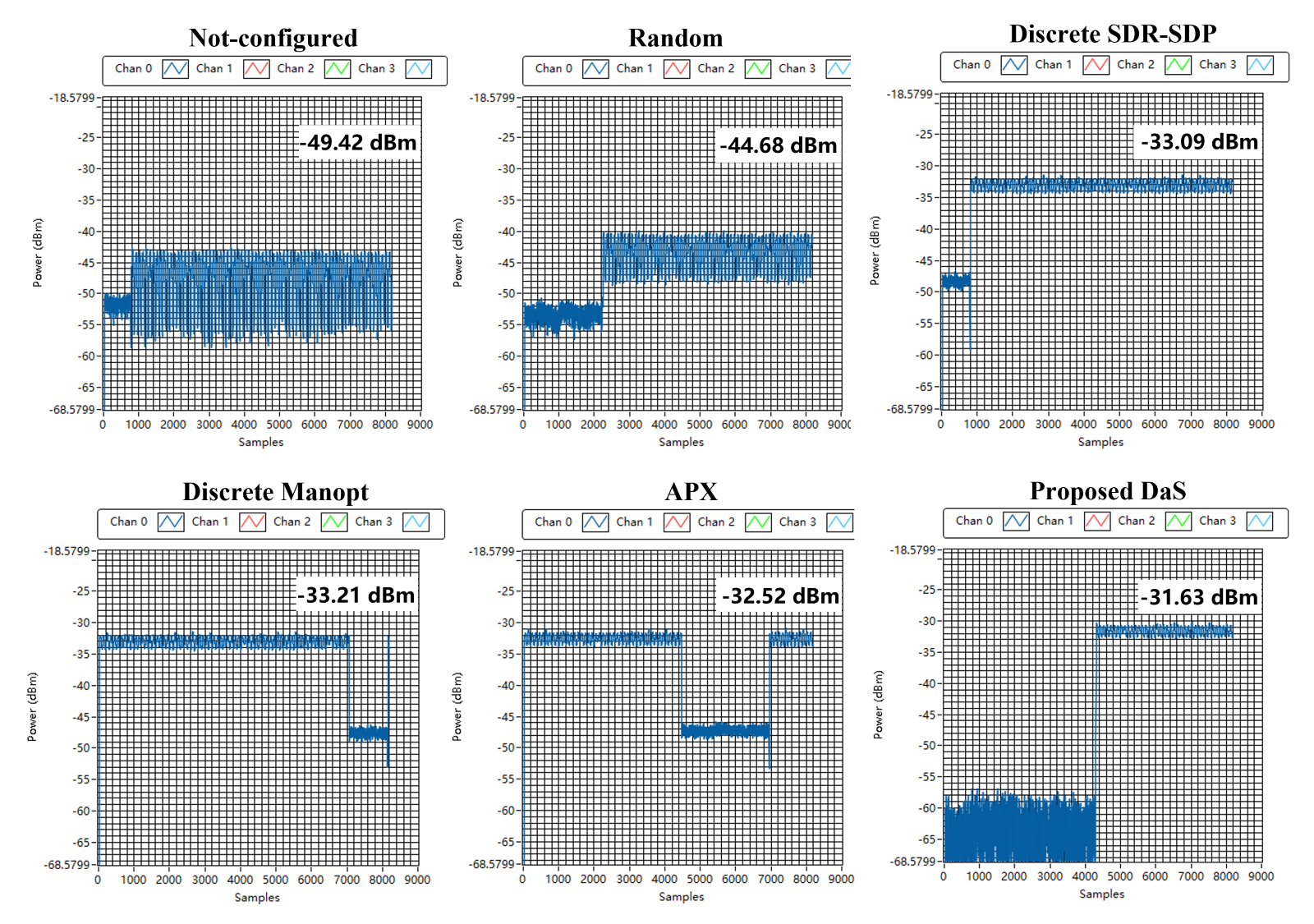}
\vspace{-0.4cm}
\caption{The received signal power measurements for different approaches in a 1-bit RIS-aided communication system operating at a frequency of 5.8 GHz. The Tx is located 2.1 meters in front of the RIS at an elevation angle of $\theta^{\text{Arr}}=0^{\circ}$, the Rx is positioned at an angle of $ \theta^{\text{Arr}}=45^{\circ}$,  with a distance of 2.97 meters to the RIS.}
\label{Result2.1-45}
\end{figure}

In summary, the simulation and experiment results have revealed several noteworthy findings.
\begin{itemize}
\item The proposed algorithm gives rise to the optimal solution in discrete beamforming for RIS. Both simulation and experimental results demonstrate its superior performance compared to other competing algorithms.
\item With the 1-bit quantization scheme performed, there is an approximate 3 dB reduction in received power compared to the continuous phase configurations. However, when using moderate resolution quantization, such as 4-bit and above, there is no noticeable difference between continuous and discrete phase configurations in terms of the received signal power. 
\end{itemize}

\section{Conclusion}\label{Section9}

This paper presents a novel framework dedicated to maximizing the $\ell_p$-norm with discrete uni-modular variable constraints. Our proposed alternating inner product maximization approach represents the first post-rounding lifting method capable of mitigating performance degradation caused by discrete quantization. Numerical simulations demonstrate superior performance in terms of SNR and execution-time compared to other competing methods. Finally, we validate the effectiveness of the alternating inner product maximization framework in beamforming through RISs using both numerical experiments and field trials on prototypes.

\section{Acknowledgment}

We would like to thank the anonymous referees of the initial version of this paper, who drew our attention to the reference \cite{ren2022linear}.

\bibliographystyle{IEEEtran}

\bibliography{Reference}

\end{document}